\documentclass[11pt]{amsart}

\usepackage{amsmath}
\usepackage{amssymb}
\usepackage{amsthm}
\usepackage{amscd, amsfonts}
\usepackage[dvips]{epsfig}
\usepackage{verbatim}

\usepackage{epsf}
\usepackage[usenames]{color}
\usepackage[bookmarksnumbered,pdfpagelabels=true,plainpages=false,colorlinks=true,
            linkcolor=black,citecolor=black,urlcolor=blue]{hyperref}

\newcommand{\four}{4}
\newcommand{\al}{\alpha}

\newcommand{\ga}{\gamma}
\newcommand{\Ga}{\Gamma}
\newcommand{\de}{\delta}

\newcommand{\la}{\lambda}
\newcommand{\La}{\Lambda}
\newcommand{\si}{\sigma}

\newcommand{\Om}{\Omega}

\newcommand{\cA}{\mathcal{A}}

\newcommand{\cF}{\mathcal{F}}

\newcommand{\cS}{\mathcal{S}}
\newcommand{\cV}{\mathcal{V}}
\newcommand{\RR}{\mathbb R}

\newcommand{\Ts}{T_{\text{st}}^{g_0}}
\newcommand{\Tsg}{T_{\text{st}}^g}
\newcommand{\ve}{\varepsilon}
\newcommand{\Length}{length}
\newcommand{\Mass}{mass}

\newcommand{\rar}{\rightarrow}

\newcommand{\vol}{\operatorname{vol}}
\newcommand{\Spec}{\operatorname{Spec}}
\theoremstyle{plain}
\newtheorem{theorem}{Theorem}[section]

\newtheorem{prop}[theorem]{Proposition}

\newtheorem{conjecture}[theorem]{Conjecture}

\theoremstyle{definition}
\newtheorem{rema}[theorem]{Remark}

\numberwithin{equation}{section}

\def\Xint#1{\mathchoice
{\XXint\displaystyle\textstyle{#1}}%
{\XXint\textstyle\scriptstyle{#1}}%
{\XXint\scriptstyle\scriptscriptstyle{#1}}%
{\XXint\scriptscriptstyle\scriptscriptstyle{#1}}%
\!\int}
\def\XXint#1#2#3{{\setbox0=\hbox{$#1{#2#3}{\int}$}
\vcenter{\hbox{$#2#3$}}\kern-.5\wd0}}
\def\dashint{\Xint-}

\title[Boundedness of effective potentials]{On the 
boundedness of effective potentials arising from string compactifications}

\author[Disconzi]{Marcelo M. Disconzi}
\address{Department of Mathematics\\
Stony Brook University\\ Stony Brook, NY 11794, USA}
\email{disconzi@math.sunysb.edu}

\author[Douglas]{Michael R. Douglas}
\address{Simons Center for Geometry and Physics\\
Stony Brook University\\ Stony Brook, NY 11794, USA \\
I.H.E.S., Le Bois-Marie, Bures-sur-Yvette, 91440 France}
\email{douglas@max2.physics.sunysb.edu}

\author[Pingali]{Vamsi Pingali}
\address{Department of Mathematics\\
Stony Brook University\\ Stony Brook, NY 11794, USA}
\email{Vamsi.Pingali@sunysb.edu}

\begin{document}

\begin{abstract}
We study effective potentials coming from compactifications of string theory.
We show that, under mild assumptions, such potentials are bounded from below in four dimensions, 
giving an affirmative answer to a conjecture proposed 
by the second author in 
\cite{Do}. We also derive some sufficient
conditions for the existence of critical points. All proofs and mathematical hypotheses
are discussed in the context of their relevance to the physics of the problem.
\end{abstract}

\maketitle

\tableofcontents

\section{Introduction}

A decade after Einstein proposed his theory of general relativity, Kaluza and Klein showed
that, by postulating an extra dimension of space, one could obtain a unified theory
of gravity and electromagnetism.  In modern terms, one takes space-time to be $5$-dimensional,
an $S^1$ fibered over observable $4$-dimensional space-time.  The $5$-dimensional metric can then be 
decomposed into a $4$-dimensional metric, a metric on $S^1$, and a one-form on $4$-dimensional space-time.
The one-form can be identified with a $U(1)$ connection, and the $5$-dimensional Einstein action becomes a 
$4$-dimensional Einstein-Maxwell action, with additional terms depending on the metric on $S^1$.

After a long period of obscurity, the ideas of Kaluza and Klein regained popularity in the 80's after the
realization that supergravity and especially superstring theory make sense in more than $4$ space-time
dimensions.  Superstring theory requires $10$-dimensional space-time \cite{Pol, GSW}, while maximally
supersymmetric supergravity and its quantum version called ``M theory'' make sense in $11$-dimensional space-time.
In both cases, one makes contact with standard $4$-dimensional physics by postulating that the 
extra dimensions form a small $n=6$ or $7$-dimensional compact manifold $M$.

A primary goal of the physics work on these compactifications is to derive an effective action in $4$ dimensions.
This is a functional of the $4$-dimensional metric and whatever additional data parametrise the extra
dimensions -- its metric, and the other fields of supergravity or superstring theory -- taken as functions on
$4$ dimensional space-time.  Critical points of this effective action (in the usual sense of a variational principle)
correspond to critical points of the original higher-dimensional supergravity or superstring action.  

The simplest case is to restrict to {\four}-dimensional maximally symmetric space-times (Minkowski, AdS, dS)
with all other fields constant in \four\ dimensions.
In this case, the effective action reduces to an effective potential, a functional of the metric and other
fields on $M$.   Physically, this potential is the energy of the {\four}-dimensional vacuum and
thus considerations of stability apply -- one is particularly interested in local minima of the effective potential,
and one has physical arguments that the effective potential is bounded below. 
This idea was turned into a precise mathematical conjecture by one of the authors in \cite{Do}, which we now describe.

We consider compactification on an $n = D - d$-dimensional compact manifold $M$ to
$d$-dimensional maximally symmetric space-time (Minkowski, AdS, dS), with $D=10$ and $d=4$ being 
the case of most interest. In the $D$-dimensional space, we consider General Relativity coupled
to matter, the latter being encoded as usual in a set (possibly empty) of field strengths $F^{(p)}$,
$p=1,\dots$ (these are curvature terms, with the standard curvature of the Yang-Mills functional being
the canonical example). After compactifying we end up with an effective potential $\cV$ which
is completely determined by quantities living in the compact $n$-dimensional manifold $M$:
\begin{gather}
 \cV = \frac{1}{2} \int_M v^2 \Big (-R_g + \frac{1}{2}\sum_{p=1}^L |F^{(p)}|_g^2 
-\frac{1}{2}T_{\text{st}} \Big )dV_g - \frac{3}{2} \int_M |\nabla_g v|_g^2 dV_g 
+ \frac{1}{2}\al \Big ( \frac{1}{G_N} - \int_M v^{2-\frac{4}{d}} dV_g \Big )
\label{effective_potential}
\end{gather}
where $G_N$ is the $d$-dimensional Newton's constant, $g$ is a metric on $M$, $v$ is a 
positive function (the so-called warp factor), $F^{(p)}$ are field strengths as mentioned before and 
$T_{\text{st}}$ is a function on $M$ which represents the non-classical sources
present in superstring theory; $\al\in\RR$ is a Lagrange multiplier, in the sense that 
its variational equation enforces a constraint (see section \ref{results} below and \cite{Do} for 
a more detailed discussion).

Exploring the physics of the problem and particular examples, one is lead to the following
\begin{conjecture} \cite{Do}: Consider a conformal class of metrics $g = e^{2\varphi}g_0$ on 
an $n$-dimensional manifold $M$; 
then the functional (\ref{effective_potential}) evaluated at its critical points $\frac{\de \cV}{\de v} =0$, considered
as a function on the space of all conformal factors $\varphi$ with fixed warped volume and volume:
\begin{gather}
\int_M v^{2-\frac{4}{d}} dV_{g} = C_1,~~~ \int_M dV_{g} = C_2
\nonumber
\end{gather}
and all $F$, is bounded below.
\label{main_conjecture}
\end{conjecture}
\noindent
We will prove this conjecture under mathematically precise hypotheses as
theorem \ref{main_theorem}. 

There is a close analogy to the Positive Mass theorem \cite{SY2, SY3, W}
and its generalisations like the Penrose inequality \cite{Mar, Wa}.
The Positive Mass theorem states that, asymptotically flat space-times
which satisfy the constraint equations of general relativity have non-negative energy.
It has been generalised to space-times which are asymptotically anti-de Sitter with a given
cosmological constant $\Lambda<0$ in \cite{CH, Wang}, and this is very much like a lower bound $\cV\ge\Lambda$.
While the present conjecture concerns a more restricted class of space-times,
in return it does not assume any asymptotic value
for $\Lambda$, only that $d$-dimensional space-time is anti-de Sitter for some $\Lambda$.

These results also naturally
relate to problems in Conformal Geometry. In particular the sign of the Yamabe invariant of $M$ 
\cite{LP, SY1} plays a role
in our criteria for the existence of critical points (see section \ref{existence_critical}), and it would be interesting to explore
possible deeper connections between the present work and the Yamabe problem.

Throughout the paper
we try to explain the role played by our hypotheses and statements in the physics of the problem.
Additional physics background can be found in the review
\cite{Douglas:2006es}.

\section{Statement of the results\label{results}}
For the rest of this section we will assume that $d=4$.
Let us rewrite conjecture \ref{main_conjecture} 
in a form more suitable for our goals.

A positive function $v$ is a critical point $\frac{ \de \cV}{\de v} = 0$ of the 
functional $\cV$ if and only if it satisfies\footnote{Notice that here a sign convention opposite to that 
in \cite{Do} is used; see the appendix for notation and conventions.}
\begin{gather}
 \Delta_g v + (- \frac{1}{3}R_g + \frac{1}{6} F_g - \frac{1}{6} \Tsg ) v = \frac{\al}{6},~~ v>0,
\label{eq_v_al}
\end{gather}
for some real number $\alpha$, where $F_g = \sum_{p=1}^L |F^{(p)}|_g^2$. $T_{\text{st}}$ will be allowed to depend on the metric
and therefore this dependence has been written explicitly. See equation (\ref{T_transformation}) and 
the discussion that follows.

Let us first recall the sense in which the parameter $\al$ is a Lagrange multiplier. Since $d=4$, the first 
constraint in conjecture \ref{main_conjecture} is simply $\int_M v \,dV_g = C_1$. As explained in \cite{Do},
the constant $C_1$ is the prefactor of the $4$-dimensional Einstein action (the integrated scalar curvature),
which is a physically measurable constant (essentially, the inverse of Newton's constant).
Thus, we write
\begin{gather}
\int_M v\, dV_g = \frac{1}{G_N}
\label{const_G_N} 
\end{gather}
in terms of a fixed constant $G_N$.\footnote{One could set $G_N=1$ and thus choose the physical unit of length.
Also, we have left out a conventional factor of $16\pi$.}
Notice that solutions to (\ref{eq_v_al}) need not automatically satisfy 
(\ref{const_G_N}).
However, if there exists a solution $v_0$ of (\ref{eq_v_al}) satisfying (\ref{const_G_N}), 
then, all other solutions
will automatically satisfy this condition. Indeed, notice that every solution of (\ref{eq_v_al})
is of the form 
\begin{gather}
v = v_0 + w ,
\label{affine}
\end{gather}
where $w$ is a solution to the homogeneous equation associated with
(\ref{eq_v_al}). Since $w$ must satisfy $\int_M w\, dV_{g} = 0$ (see the proof of $(1)$ in theorem \ref{main_theorem}),
if follows that
\begin{gather}
\int_M v\, dV_g = \int_M v_0 \,dV_g = \frac{1}{G_N}.
\label{orthogonality}
\end{gather}
Now, (\ref{eq_v_al}), being a linear equation, has the property that, a solution $v$ with a particular 
value of $\al$ can be rescaled to a solution $\lambda v$ for another value $\lambda\al$. Therefore we can
choose $\al$ such that (\ref{const_G_N}) holds. Unless stated otherwise, from now on we assume that critical 
points are always tuned to satisfy condition (\ref{const_G_N}).

Evaluating the functional at $v$ and using the constraints yields
\begin{gather}
 \cV = \frac{\al}{4G_N} .
\label{func_const}
\end{gather}
Because of (\ref{const_G_N}), the dependence of $\cV$ on $\varphi$ and $v$ is through the Lagrange
multiplier $\al$, and we now seek to write this dependence in a more explicit fashion.

In light of (\ref{func_const}), in order to prove conjecture \ref{main_conjecture} we need to
only discuss the case where 
$\al$ is negative. So, in the rest of this section we make that 
assumption (some features of the case $\al \geq 0$ are discussed in sections \ref{existence_critical} and 
\ref{d_not_4}).
Defining $u$ by $v = \frac{|\al|}{6} u$, then
it is seen to satisfy
\begin{gather}
P_g u \equiv \Delta_g u + (- \frac{1}{3}R_g + \frac{1}{6} F_g - \frac{1}{6} \Tsg ) u = -1.
\label{eq_u_g}
\end{gather}
Using (\ref{const_G_N}) to express $\al$ in terms of $u$, one has
\begin{gather}
 |\al| = \frac{6}{G_N\int_M u dV_g}.
\label{al_in_terms_u}
\end{gather}
Writing a general metric $g$ in the conformal class in terms of a fixed
background metric $g_0$, $g = e^{2\varphi} g_0$, and using 
(\ref{affine}), (\ref{orthogonality}), and (\ref{al_in_terms_u})  into (\ref{func_const}) finally gives
\begin{gather}
\cV = - \frac{6}{4G_N^2 \int_M e^{n\varphi} u\, dV_{g_0} } .
\label{funct_final_form}
\end{gather}

By the solution of the Yamabe problem
\cite{Au1,  S1, Tr, Ya}, the metric $g_0$ can be assumed to have constant scalar curvature $R_{g_0}$, and
henceforth we do so.
This will be  positive, negative, or zero according
to the sign of the Yamabe invariant\footnote{We use the term ``Yamabe invariant'' to denote the
invariant of the conformal class, whereas the supremum 
over all conformal classes is called ``Topological Yamabe invariant".}
of $(M,g_0)$. It will be convenient to express all quantities in terms of this fixed background metric.

Under $g = e^{2\varphi} g_0$,
\begin{gather}
 R_g = e^{-2\varphi} \big (-2(n-1) \Delta_{g_0} \varphi - (n-1)(n-2) |\nabla_{g_0}\varphi|_{g_0}^2 + R_{g_0} \big), 
\label{scalar_transformation}
\end{gather}
and 
\begin{gather}
 \Delta_g u = e^{-2\varphi}  \big (\Delta_{g_0} u + (n-2) \langle \nabla_{g_0} \varphi, \nabla_{g_0} u \rangle_{g_0} \big ).
\label{Laplacian_transformation}
\end{gather}

Recall that, if $F^{(p)}$ is a $p$-form, then
\begin{align}
\begin{split}
 |F^{(p)}|_g^2 & = g^{\mu_1 \nu_1} \cdots g^{\mu_p \nu_p} F_{\mu_1 \cdots \mu_p} F_{\nu_1 \cdots \nu_p} \\
& = e^{-2p\varphi} g_0^{\mu_1 \nu_1} \cdots g_0^{\mu_p \nu_p} F_{\mu_1 \cdots \mu_p} F_{\nu_1 \cdots \nu_p} \\
& = e^{-2p\varphi} |F^{(p)}|_{g_0}^2 .
\label{gauge_fields_p}
\end{split}
\end{align}
Hence, the gauge fields expressed in terms of the metric $g_0$ and $\varphi$ become
\begin{gather}
\sum_{p=1}^L |F^{(p)}|_g^2 = \sum_{p=1}^L e^{-2p\varphi} |F^{(p)}|_{g_0}^2 .
\label{gauge_fields_transformation}
\end{gather}

Now we need to ask how $\Tsg$ transforms under $g = e^{2\varphi } g_0$.
The basic example of $\Tsg$ in string theory is the so-called orientifold plane, which is supported on
a submanifold \cite{CHSW, GKP}. Another common example is a Chern-Simons or topological term,
as used in \cite{Becker,GKP}.  What is important for our problem is its dependence on the conformal factor.
Thus, we will take $\Tsg$ to be a function or  a generalised function on $M$, and postulate a homogeneous
dependence on the conformal factor,
\begin{gather}
 \Tsg = e^{-2\beta \varphi} \Ts ,
\label{T_transformation}
\end{gather}
which is consistent with (\ref{gauge_fields_transformation}).
In the main result of this paper, 
theorem \ref{main_theorem}, we will consider different situations which 
will allow for different choices of $\beta$.
The first two cases, $\beta=0$ and $\beta = 1$, will be treated together, and in fact 
the proof in these cases works for any $0 \leq \beta \leq 1$ .  Another case of special interest
is $\beta=n/2$ which is appropriate for a delta function source.  Finally, 
a variation on this condition which includes orientifolds and many other cases is to make
the simple assumption
\begin{gather}
 \int_M \Tsg \, dV_g = \int_M \Ts \, dV_{g_0}, ~~\text { for all } g = e^{2\varphi} g_0.
\label{transformation_orientifold}
\end{gather}
See sections \ref{discussion} and \ref{examples} for a more detailed discussion on the relevance of each hypothesis 
for the physics of string compactifications.

Upon combining (\ref{scalar_transformation}), (\ref{Laplacian_transformation}), (\ref{gauge_fields_transformation})
and (\ref{T_transformation}), equation (\ref{eq_u_g}) becomes
\begin{gather}
 M_{g_0} u =  \Delta_{g_0} u + (n-2) \langle \nabla_{g_0} \varphi, \nabla_{g_0} u \rangle_{g_0} + U\,  u = -e^{2\varphi} ,
\label{eq_u} 
\end{gather}
with
\begin{gather}
U \equiv \frac{2}{3}(n-1) \Delta_{g_0} \varphi + \frac{1}{3}(n-1)(n-2) |\nabla_{g_0} \varphi|^2_{g_0} 
- \frac{1}{3}R_{g_0} + \cF(\varphi) - \frac{1}{6} e^{2(1-\beta) \varphi} \Ts 
\label{def_U}
\end{gather}
and
\begin{gather}
 \cF(\varphi) = \frac{1}{6}  \sum_{p=1}^L e^{2(1-p)\varphi}  |F^{(p)}|_{g_0}^2 .
\label{cF_def}
\end{gather}
We will 
not need the specific form of $\cF$, it being enough for our proof to notice that $\cF \geq 0$.\footnote{
This is actually not manifest in supergravity as there are Chern-Simons and other non-quadratic terms,
but it is shown in Ref. \cite{Douglas:2010rt}.}

Let us now restate conjecture \ref{main_conjecture} as a theorem that will be proven in the subsequent sections.
After stating the results we discuss the physical meaning of the hypotheses some consequences.

\begin{theorem}
Let $(M, g_0)$ be a compact orientable $n$-dimensional Riemannian manifold without boundary,
let $A \in \RR_+$. Fix a collection of 
smooth $p$ forms $\{ F^{(p)} \}_{p=1}^L$, and, a function $\Ts \in C^\infty(M)$ which transforms under conformal changes as
(\ref{T_transformation}), for some $\beta \geq 0$. For any smooth 
function $\varphi$ on $M$ define  $\cF(\varphi)$ as in (\ref{cF_def}).
Define $\cA = \{ \varphi \in C^\infty(M) ~|~ \int_M e^{n\varphi} \,dV_{g_0} = A \}$, and
let $\cS \subset \cA$ be the set of $\varphi \in \cA$ such that equation (\ref{eq_u}) has a positive solution 
$u = u(\varphi,g_0, \Ts, \cF)$. Then: \\

\noindent (1) The map
$\mathfrak{F}_{g_0,\Ts,\cF}: \cS \rar \RR$ given by
\begin{gather}
 \mathfrak{F}_{g_0,\Ts,\cF}(\varphi) = - \frac{6}{4G_N^2 \int_M e^{n\varphi} u \, dV_{g_0} },
\label{functional}
\end{gather}
is well defined. \\

\noindent (2) Fix $\eta>0$ and define 
\begin{gather}
\cS_\eta = \Big \{ \varphi \in \cS ~ \Big |~ \int_M R_g \, dV_g \leq \eta \,, \text{ where } g = e^{2\varphi}g_0 \Big \} .
\nonumber
\end{gather}
Suppose that $0 \leq \beta \leq 1$, or $\beta = \frac{n}{2}$, or (\ref{transformation_orientifold}) holds. Then there 
exists a constant $K_{\eta} \in \RR$ such that
\begin{gather}
 \inf_{\varphi \in \cS_\eta } \mathfrak{F}_{g_0,\Ts,\cF}(\varphi) \geq K_{\eta}.
\label{bound_1}
\end{gather}
In particular, $\mathfrak{F}_{g_0,\Ts,\cF}$ is bounded from below if the scalar 
curvature is uniformly bounded. Moreover, if $\dim M = 2$ then
\begin{gather}
 \inf_{\varphi \in \cS} \mathfrak{F}(\varphi) \geq K,
\label{bound_dim_2}
\end{gather}
for some $K \in \RR$. \\

\noindent (3) Furthermore, if we define
\begin{gather}
\widetilde{\cS}_\eta = \Big \{ \varphi \in \cS_\eta ~ \text{ such that } 
 \int_M \big ( F_g - \Tsg \big ) \, dV_g \geq - \eta \,, \text{ where } g = e^{2\varphi}g_0 \Big \} ,
\label{condition_negative_type}
\end{gather}
then there exists a constant $K_{\eta} \in \RR$ such that
\begin{gather}
 \inf_{\varphi \in \widetilde{\cS}_\eta } \mathfrak{F}_{g_0,\Ts,\cF}(\varphi) \geq K_{\eta},
\label{bound_1_int}
\end{gather}
for any value of $\beta$ in (\ref{T_transformation}).

\label{main_theorem}
\end{theorem}
A different approach to the theorem
would be to impose conditions directly on the functions $\varphi$ appearing
in the conformal factor $e^{2\varphi}$. In this regard we prove:
\begin{prop}
Assume the same hypotheses of Theorem \ref{main_theorem}. Suppose that $0\leq \beta \leq 1$ and 
that $\dim(M) \geq 3$. If $B_\eta(0)$ is 
the ball of radius $\eta$ in $L^1(M,g_0)$ then
there exists a constant $K_{\eta} \in \RR$ such that
\begin{gather}
 \inf_{\varphi \in \cS \cap B_\eta(0)} \mathfrak{F}_{g_0,\Ts,\cF}(\varphi) \geq K_{\eta}.
\label{bound_2}
\end{gather}
\label{little_prop}
\end{prop}

\subsection{Discussion of the hypotheses.\label{discussion}}

Let us make some comments on the hypotheses and the content of theorem \ref{main_theorem}.
We first discuss the physics of  (\ref{T_transformation}) and (\ref{transformation_orientifold}).

The ``non-classical'' source terms $T_{\text{st}}$ differ from the usual stress-energy tensor of general
relativity in two ways.  First, they are generally associated to quantum effects and anomalies in
string theory and M theory, as in the anomaly cancellation terms of the heterotic string called
on in \cite{CHSW}.  Second, and more importantly for us, they violate the positive energy condition.
This allows finding compactifications to Minkowski space-time even in the presence of other positive
contributions to the energy, again as first seen in \cite{CHSW}.  But it means that any positive (or bounded
below) energy theorem in string/M theory will require placing some condition on $T_{\text{st}}$, which replaces the
positive energy condition.  Although
the physical consistency of string/M theory implies that some such condition exists, and
one can see some of its features in examples, at
present no precise and sufficiently general statement of the condition has been proposed.  Thus we put forward 
(\ref{T_transformation}) and (\ref{transformation_orientifold}) as candidates, which suffice to prove our main results.

Some of the $\beta$ values of interest 
are the ``topological'' case referred to above, namely, $\beta = 0$, 
and $\beta = \frac{n}{2}$, which is how $\Tsg$ transforms when it is a delta function, since
\begin{gather}
1 =  \int_M \de_g \, dV_g = \int_M  e^{-n \varphi}  \de_{g_0} e^{n\varphi} \, dV_{g_0} =  \int_M \de_{g_0} \, dV_{g_0}.
\nonumber 
\end{gather}
Notice that $\beta = \frac{n}{2}$ is also the case where $\Tsg$ scales in the same way as 
in (\ref{gauge_fields_p}) with $p=\frac{n}{2}$, which is when the the gauge fields 
have a conformally invariant stress tensor. Another natural case 
is the ``naive'' choice $\beta = 1$, i.e., declaring that $\Tsg$ transforms in the same way as the metric.

Yet another important example is the orientifold plane mentioned above. In this case, $\Tsg$ 
is  supported on a closed submanifold $N \subset M$ in the sense that
\begin{gather}
 \Tsg \equiv 0 ~~\text{ on }~~ M\backslash N, ~~\text{ and } ~~\int_M \Tsg \, dV_g = \int_N \Tsg \, d\Sigma_g = 1,
\label{orientifold_Tsg}
\end{gather}
where $d\Sigma_g$ is the induced volume element on $N$. The simplest way to model 
(\ref{orientifold_Tsg}) is to have $\Tsg = \de_g(N)$, i.e., the measure assigning one to $N$ and
zero to $M\backslash N$, although choices like $\Tsg = \chi_N$, or some smooth approximation of it, 
can also be considered ($\chi_E$ being the characteristic function of a set $E$).
If $N$ is $k$-dimensional, then $d\Sigma_g = e^{k\varphi}\,d\Sigma_{g_0}$, and hence by analogy with the 
delta function, we require that in the orientifold case the string term transforms as
\begin{gather}
 \Tsg = e^{-k\varphi} \Ts.
\nonumber
\end{gather}
This corresponds to $\beta = \frac{k}{2}$, where $k$ is the dimension of the submanifold supporting
$\Ts$. Notice that $\beta = \frac{n}{2}$ and an orientifold-like string term can 
both be encoded in the assumption (\ref{transformation_orientifold}).

Finally, even if neither of these conditions were to hold, we 
can prove the result under the global condition (\ref{condition_negative_type}),
which states that, any negative contribution of $\Tsg$ is compensated by the energy of matter and fluxes,
as comes out of many analyses and discussions \cite{GKP,MN,WSD}.

Turning to mathematical questions, the first and most obvious is whether the set $\cS$ is empty or not. 
We address this question in section \ref{existence_critical}, where we provide conditions
for the existence of solutions to equation (\ref{eq_u}).  We will also see reasons to think that
$\cS\neq \cA$ in general.  Presumably this is because string/M theory effects cannot be neglected in
these cases, but not much is known about this.

That the map $\mathfrak{F}_{g_0,\Ts,\cF}$ is well-defined is not a surprise -- by Eq. (\ref{orthogonality}),
$\mathfrak{F}_{g_0,\Ts,\cF}$ does not depend on elements in the kernel of $P_g$. 

In (\ref{bound_1}), the infimum is taken not only over all conformal factors satisfying the 
original constraint $\int_M e^{n\varphi} \,dV_{g_0} = A$, but also obeying the extra integral bound 
$\int_M R_g \, dV_g \leq \eta$.  

Curvature bounds are physically appropriate whenever one studies quantum gravity and string/M theory 
in the language of general relativity.  These theories have a preferred scale of length below which a description
in terms of classical space-time breaks down, the Planck length $L_{pl}$ for quantum gravity and M theory,
and the string length $L_s$
for string theory.  Let $L=\max(L_{pl},L_s)$, then a description in terms of a space-time metric satisfying
Einstein's equations will generally only be valid when the curvature of the metric is much less than $1/L^2$.
This would apply to every component of the curvature as well as its derivatives, and the $p$-form field strengths
$F^{(p)}$ and their derivatives.  In this sense, the physics naturally places stronger conditions, such as 
uniform curvature 
bounds\footnote{We recall that curvature bounds 
have been successfully employed to study the long-time existence of Einstein equations \cite{An, KR}.}.  
As we explain below, the proof with these conditions is an 
elementary consequence of the more general assumptions that are adopted here.

Mathematically, since all constraints in conjecture \ref{main_conjecture}
are integral conditions, one would prefer, even if merely for aesthetic reasons, to 
have any additional conditions to be integral as well, as in theorem \ref{main_theorem}.

The bound on the scalar curvature is also a very reasonable
geometrical assumption, in the sense
that curvature bounds are commonplace in Riemannian Geometry (see e.g. \cite{Au2, CE,  SY1}).
Finally, we point out that $\int_M R_g \, dV_g \leq \eta$ is a natural generalization
of the situation in $2$ dimensions, where the condition is automatically 
satisfied due the the Gauss-Bonnet formula.

We should also make a comment about our assumptions of regularity. 
In order to avoid technicalities that would 
obfuscate the main ideas, all our functions are assumed to be smooth.
This obviously excludes cases such as $\Ts \sim \de_{g_0}(p)$, for some $p \in M$,
or $\Ts \sim \delta_{g_0}(N)$, for some submanifold $N$.
However it will be clear that, after a suitable interpretation in the context of linear
equations involving a generalized function, the same argument works (see section \ref{examples}). 
Other regularity assumptions can also be greatly relaxed.
See for example \cite{GT, LaUl, Tr1, Tr2} for generalizations of the techniques here employed 
to conditions of less regularity.

Finally, let us address the $L^1(M,g_0)$ condition in proposition \ref{little_prop}.
A bound on the $L^1(M,g_0)$ norm of $\varphi$ seems to be a fairly standard
hypothesis if we focus exclusively on the analytical aspects of the problem.
From a more geometric perspective, it can be interpreted as follows. Even 
smooth conformal factors can approach
distributions with very bad singularities, in which case the limit metrics 
would be highly degenerate. On the other hand, singularities do occur in quantum field theory, and 
we may not want to completely rule them out by impose very strong conditions, and hence
simple $L^p$ bounds seem appropriate. Moreover, as we are dealing with compact 
manifolds, the $L^1$ norm is the weakest of all $L^p$ norms and so the weakest possible $L^p$
bound is to require $\varphi \in \cS \cap B_\eta(0)$.

At the end of the day, our hypotheses should be justified on physical grounds. 
Imposing that $\int_M \varphi \,dV_{g_0}$ is bounded above is more than appropriate, as we
know from experience that if the extra compact dimensions exist, they have to be small. On the other
hand, if $\int_M \varphi \,dV_{g_0} \rar -\infty$ then such extra dimensions would collapse. 
It would be interesting to analyze what type of physics can emerge in this setting. 
For example, we illustrate in section \ref{proof_remarks} that, if the effective
potential is not bounded from below then the singularities assume a rather specific form.

\section{Proof of theorem \ref{main_theorem}}

In this section we prove theorem \ref{main_theorem}.  We use the letter $C$
to denote several different constants that appear in the estimates. They
will never depend on $\varphi$, $u$ or $i$, where $i$ indexes a sequence in $\cS$ (see 
proof below). For $0 \leq \beta \leq 1$, the most natural choices are $\beta =0,1$, as
we stated in the theorem. But since the proof works without modifications
for values of $\beta$ between zero and one, we will write 
 $0 \leq \beta \leq 1$ explicitly in some passages of the proof in order to stress this fact. \\

\noindent \emph{Proof of theorem \ref{main_theorem}-(1)}:
First we need to show that the map
$\mathfrak{F}_{g_0,\Ts,\cF}$ is well-defined. For a given $\varphi \in \cS$, let $w$ belong to 
the kernel of $P_g$, or equivalently to the kernel of $M_{g_0}$, where
$P_g$ and $M_{g_0}$ have been defined in (\ref{eq_u_g}) and (\ref{eq_u}), respectively. 
By the Fredholm alternative (see appendix), $w$ is then $L^2$ orthogonal to the image of $P_g$, and since
$-1$ belongs to the image by hypothesis,
\begin{gather}
 - \int_M w \, dV_g = - \int_M w e^{n\varphi} \,dV_{g_0} = 0 .
\nonumber
\end{gather}
Now, if $u_1$ and $u_2$ are two different solutions of (\ref{eq_u_g}), then $u_1 - u_2$ belongs
to the kernel of $P_g$ and
\begin{gather}
 \int_M  e^{n\varphi} u_1 \, dV_{g_0} =
\int_M  e^{n\varphi} u_2 \, dV_{g_0} +  \int_M  e^{n\varphi}  (u_1 - u_2) \, dV_{g_0}
=\int_M  e^{n\varphi} u_2 \, dV_{g_0} ,
\nonumber
\end{gather}
showing that the map is well defined. \\
 
\noindent \emph{Proof of theorem \ref{main_theorem}-(2), $0 \leq \beta \leq 1$}:
Here we prove part \emph{(2)} of the theorem for $\beta$ between zero and one. 
A separate proof for the case 
$\beta=\frac{n}{2}$ and when (\ref{transformation_orientifold}) holds will be provided for the reasons explained
in section \ref{results}.

We start deriving an useful inequality. Notice that $v>0$, and so
is $u$. We are therefore allowed to divide equation
(\ref{eq_u}) by $u$. Doing so and integrating by parts yields
\begin{gather}
\begin{split}
 \int_M \frac{| \nabla_{g_0} u|^2_{g_0} }{u^2} \, dV_{g_0} 
+ \frac{1}{3}(n-1)(n-2) \int_M |\nabla_{g_0} \varphi |^2 \, dV_{g_0} 
+ (n-2) \int_M \frac{1}{u} \langle \nabla_{g_0} \varphi, \nabla_{g_0} u \rangle_{g_0} \, dV_{g_0} 
\label{basic_identity} \\
+ \int_M \frac{e^{2\varphi}}{u} \, dV_{g_0} + \int_M \cF \, dV_{g_0} 
- \frac{1}{6} \int_M e^{2(1-\beta) \varphi} \Ts \, dV_{g_0} = \frac{1}{3} R_{g_0} \vol_{g_0}(M) .
\end{split}
\end{gather}
The Cauchy inequality-with-epsilon (see appendix) gives
\begin{gather}
\begin{split}
 \int_M \frac{1}{u} \langle \nabla_{g_0} \varphi, \nabla_{g_0} u \rangle_{g_0} \, dV_{g_0} 
\geq - \frac{\ve}{2} \int_M |\nabla_{g_0} \varphi |^2 \, dV_{g_0} 
- \frac{1}{2\ve} \int_M \frac{| \nabla_{g_0} u|^2_{g_0} }{u^2} \, dV_{g_0} , 
\nonumber
\end{split}
\end{gather}
so,
\begin{gather}
\begin{split}
\int_M \frac{| \nabla_{g_0} u|^2_{g_0} }{u^2} \, dV_{g_0} 
+ \frac{1}{3}(n-1)(n-2) \int_M |\nabla_{g_0} \varphi |^2 \, dV_{g_0} 
+ (n-2) \int_M \frac{1}{u} \langle \nabla_{g_0} \varphi, \nabla_{g_0} u \rangle_{g_0} \, dV_{g_0} 
\label{cauchy_1}\\
\geq \Big( 1 - \frac{n-2}{2\ve} \Big) \int_M \frac{| \nabla_{g_0} u|^2_{g_0} }{u^2} \, dV_{g_0} 
+ (n-2) \Big ( \frac{n-1}{3} - \frac{\ve}{2} \Big ) \int_M |\nabla_{g_0} \varphi |^2 \, dV_{g_0} .
\end{split}
\end{gather}
Since $\frac{n-2}{2} < \frac{2(n-1)}{3}$ one can choose $\ve$ such that $\ve > \frac{n-2}{2}$ and
$\ve < \frac{2(n-1)}{3}$. Hence (\ref{basic_identity}) and (\ref{cauchy_1}) combine to give
\begin{gather}
\begin{split}
\frac{1}{3} R_{g_0} \vol_{g_0}(M)  \geq c_1 \int_M \frac{| \nabla_{g_0} u|^2_{g_0} }{u^2} \, dV_{g_0} 
+ c_2 \int_M |\nabla_{g_0} \varphi |^2 \, dV_{g_0} 
\label{basic_inequality} \\
+ \int_M \frac{e^{2\varphi}}{u} \, dV_{g_0} + \int_M \cF \, dV_{g_0} 
- \frac{1}{6} \int_M e^{2(1-\beta) \varphi} \Ts \, dV_{g_0} ,
\end{split}
\end{gather}
where $c_1$ and $c_2$ are positive constants. 

To prove (\ref{bound_1}), assume the result is not true. Then there exists
a sequence $\varphi_i \in S_\eta$, with corresponding solutions $u_i$ of (\ref{eq_u}), such that
$\mathfrak{F}_{g_0,\Ts}(\varphi_i) \rar - \infty$ as $i \rar \infty$, and therefore
\begin{gather}
\int_M e^{n\varphi_i} u_i \, dV_{g_0} \rar 0^+.
\label{limit_zero} 
\end{gather}
We will suppress the subscript $i$
for notational convenience, but all limits are to be understood as the limit when $i \rar \infty$.

Start noticing that 
$\int_M e^{n\varphi} \,dV_{g_0} = A$ implies that there exists a constant $C>0$ such that
\begin{gather}
 \int_M e^{2(1-\beta)\varphi} \, dV_{g_0} \leq C.
\nonumber
\end{gather}
This is obvious when $\beta = 1$ and for $\beta < 1$, H\"older's inequality gives
\begin{align}
\begin{split} 
\int_M e^{2(1-\beta)\varphi} \, dV_{g_0} & \leq \Big ( \vol_{g_0}(M) \Big )^\frac{n}{n-2(1-\beta)}
\parallel e^{ 2(1-\beta)\varphi}\parallel_{L^\frac{n}{2(1-\beta)}(M,g_0)} \\
& = \Big ( \vol_{g_0}(M) \Big )^\frac{n}{n-2(1-\beta)} \Big (\int_M e^{n\varphi} \, dV_{g_0} \Big )^\frac{2(1-\beta)}{n}.
\nonumber
\end{split}
\end{align}
Hence
\begin{align}
\Big |  \int_M e^{2(1-\beta) \varphi} \Ts \, dV_{g_0} \Big | & 
\leq \sup_{x\in M} | \Ts(x) | \int_M e^{2(1-\beta) \varphi}  \, dV_{g_0}
\leq C \sup_{x\in M} | \Ts(x) | .
\nonumber
\end{align}
Combining this with the fact that the first, second and fourth integrals on the 
right hand side of (\ref{basic_inequality}) are non-negative, produces
\begin{gather}
\begin{split}
C + \frac{1}{3} R_{g_0} \vol_{g_0}(M)
\geq \frac{1}{6} \int_M e^{2(1-\beta) \varphi} \Ts \, dV_{g_0} + 
\frac{1}{3} R_{g_0} \vol_{g_0}(M) \geq 
 \int_M \frac{e^{2\varphi}}{u} \, dV_{g_0} .
\nonumber
\end{split}
\end{gather}
Or, in other words, $\int_M \frac{e^{2 \varphi}}{u} \, dV_{g_0}$ is bounded independent
of $\varphi$.
Applying Cauchy-Schwartz gives
\begin{gather}
  0 \leq \int_M e^{\frac{n+2}{2}\varphi} \, dV_{g_0} = 
\int_M e^{\frac{n}{2}\varphi}\sqrt{u} \frac{e^{\varphi}}{\sqrt{u}} \, dV_{g_0} 
\leq \Big ( \int_M e^{n\varphi} u \, dV_{g_0} \Big )^\frac{1}{2} 
\Big ( \int_M \frac{e^{2 \varphi}}{u} \, dV_{g_0} \Big )^\frac{1}{2} .
\nonumber
\end{gather}
As the last term is bounded, from (\ref{limit_zero}) it follows that
\begin{gather}
\parallel e^\varphi \parallel_{L^\frac{n+2}{2}(M,g_0)}^\frac{n+2}{2} =  
\int_M e^{\frac{n+2}{2}\varphi} \rar 0 .
\nonumber
\end{gather}
By the H\"older inequality it then follows that $\parallel e^\varphi \parallel_{L^s(M,g_0)} \rar 0$ 
for any $1 \leq s \leq \frac{n+2}{2}$ . 
Now we will bootstrap to obtain this result for $s < n$. For any $p>1$, letting 
$\frac{1}{p}+\frac{1}{q} =1$ as usual, we see that: 
\begin{gather}
 \int_M e^{s\varphi} \, dV_{g_0} = 
\int_M e^{(s-\frac{n}{p})\varphi} e^{\frac{n}{p}\varphi} \, dV_{g_0}  \leq
\Big ( \int_M e^{q(s-\frac{n}{p})\varphi} \Big )^\frac{1}{q} \Big (\int_M e^{n \varphi} \Big )^\frac{1}{p}
= A^\frac{1}{p} \Big ( \int_M e^{\frac{p}{p-1}(s-\frac{n}{p})\varphi} \Big )^\frac{p-1}{p} 
\label{boots_int}
\end{gather}
Given $\frac{n+2}{2} < s < n$, we can always choose $p > 1$ such that 
\begin{gather}
\frac{p}{p-1}(s-\frac{n}{p}) = \frac{n+2}{2},
\label{eq_s_p} 
\end{gather}
so that the right hand side of (\ref{boots_int}) goes to zero.
In fact, as $p>1$, write $p = 1 + \de$, $\de > 0$, so that (\ref{eq_s_p}) reads
\begin{gather}
 s = \frac{n+2}{2} \frac{\de}{1+\de} + \frac{n}{1+\de}.
\nonumber
\end{gather}
Then $s$ is a decreasing function of $\de$, satisfying $s \rar n$ when $\de \rar 0$ and $s \rar \frac{n+2}{2}$ as $\de \rar \infty$.
Hence we have shown:
\begin{gather}
 \int_M e^{s\varphi}\,dV_{g_0} \rar 0~\text{ for any }~1\leq s < n.
\label{boostrap}
\end{gather}

Now multiply (\ref{eq_u}) by $\frac{e^{(n-2)\varphi}}{u}$, integrate, and integrate by parts the Laplacian terms.
Then terms involving $\frac{1}{u}\langle \nabla_{g_0} \varphi, \nabla_{g_0} u \rangle_{g_0}$ cancel out and
we obtain
\begin{gather}
\begin{split}
 \frac{1}{3}(n-1)(n-2)\int_M e^{(n-2)\varphi} |\nabla_{g_0} \varphi|^2 \,dV_{g_0} =
\int_M \frac{e^{n\varphi}}{u} \, dV_{g_0} + \int_M \frac{|\nabla_{g_0} u |^2}{u^2} e^{(n-2)\varphi} \, dV_{g_0} 
\label{n_2_varphi} \\
 -\frac{1}{3} \int_M e^{(n-2)\varphi} R_{g_0} \,dV_{g_0} 
+ \int_M  \cF \,dV_{g_0} -\frac{1}{6} \int_M   e^{(n-2\beta) \varphi} \Ts \,dV_{g_0} .
\end{split}
\end{gather}
If $\beta = 0$ then
\begin{gather}
\Big | \int_M   e^{(n-2\beta) \varphi} \Ts \,dV_{g_0}  \Big | =
\Big | \int_M   e^{n \varphi} \Ts \,dV_{g_0}  \Big | \leq 
\sup_{x \in M} |\Ts(x)| \int_M e^{n \varphi} \,dV_{g_0}   = A \sup_{x \in M} |\Ts(x)| ,
\nonumber
\end{gather}
which is uniformly bounded independent of $\varphi$. If $\beta > 0$ 
(and less or equal than one, see the hypotheses of the theorem)
then
\begin{gather}
\Big | \int_M   e^{(n-2\beta) \varphi} \Ts \,dV_{g_0}  \Big | =
\sup_{x \in M} |\Ts(x)| \int_M e^{(n-2\beta) \varphi} \,dV_{g_0} ,
\nonumber
\end{gather}
which goes to zero by (\ref{boostrap}). In any situation, there exists a constant $C>0$,
independent of $\varphi$ such that 
\begin{gather}
-\frac{1}{6} \int_M   e^{(n-2\beta) \varphi} \Ts \,dV_{g_0} \geq - C \sup_{x \in M} |\Ts(x)|.
\label{bound_sup_T}
\end{gather}
Evoking (\ref{boostrap}) again gives
\begin{gather}
\frac{1}{3} \int_M e^{(n-2)\varphi} R_{g_0} \,dV_{g_0} \rar 0 .
\label{scalar_goes_zero}
\end{gather}
Using (\ref{bound_sup_T}), (\ref{scalar_goes_zero}) with
\begin{gather}
 \int_M \frac{|\nabla_{g_0} u |^2}{u^2} \, dV_{g_0}\geq 0~~ \text{ and } \int_M  \cF \,dV_{g_0} \geq 0 ,
\nonumber
\end{gather}
into (\ref{n_2_varphi}) implies
\begin{gather}
 \frac{1}{3}(n-1)(n-2)\int_M e^{(n-2)\varphi} |\nabla_{g_0} \varphi |^2 \,dV_{g_0} \geq 
\int_M \frac{e^{n\varphi}}{u} \, dV_{g_0}  - C \sup_{x \in M} |\Ts(x)| + o(1),
\label{n_2_varphi_boud}
\end{gather}
where $o(1)$ denotes as usual a term that goes to zero.
Since $u>0$
and $f(x)=\frac{1}{x}$ is convex for $x>0$, Jensen's inequality (see appendix) with the measure $dV_g = e^{n\varphi}dV_{g_0}$ gives:
\begin{gather}
 \int_M \frac{e^{n\varphi }}{u} \, dV_{g_0} = \vol_g(M) \dashint_M \frac{1}{u}\,dV_g 
\geq A^2 \frac{1}{\int_M u \,dV_g} = \frac{A^2 }{\int_M e^{n\varphi }u \,dV_{g_0}} 
\label{jensen}
\end{gather}
where $\dashint_M =\frac{1}{\vol(M)} \int_M$. Therefore $\int_M \frac{e^{n\varphi }}{u}\,dV_{g_0}$ goes 
to infinity, and so 
by (\ref{n_2_varphi_boud}) we get
\begin{gather}
\int_M e^{(n-2)\varphi} |\nabla_{g_0} \varphi |^2 \,dV_{g_0} \rar \infty .
\label{limit_contrad}
\end{gather}

Now integrate (\ref{scalar_transformation}) with respect to $dV_{g} = e^{n\varphi}dV_{g_0}$ and
integrate by parts the Laplacian term to find
\begin{gather}
 \int_M R_g \, dV_g = (n-1)(n-2)\int_M e^{(n-2)\varphi} |\nabla_{g_0} \varphi |^2 \, dV_{g_0} 
+ R_{g_0} \int_M e^{(n-2)\varphi}  \, dV_{g_0} .
\nonumber
\end{gather}
The second integral approaches zero by (\ref{boostrap}), and since 
$\int_M R_g \, dV_g \leq \eta$ by hypothesis, we obtain a contradiction 
with (\ref{limit_contrad}). 
This proves (\ref{bound_1}) when $0 \leq \beta \leq 1$. \\

\noindent \emph{Proof of theorem \ref{main_theorem}-(2), $\beta = \frac{n}{2}$ or (\ref{transformation_orientifold}) holds}:
It is enough to assume (\ref{transformation_orientifold}), since this is satisfied when $\beta=\frac{n}{2}$ because
\begin{gather}
 \int_M \Tsg \, dV_g = \int_M e^{-n\varphi} \Ts e^{n\varphi} \, dV_g = 
\int_M \Ts \, dV_{g_0}.
\nonumber
\end{gather}
As before, assume the result is not true so that
\begin{gather}
 \int_M e^{n\varphi} u \, dV_{g_0} = \int_M u \, dV_g \rar 0 ,
\label{contradiction_beta_n_2}
\end{gather}
for some sequence of functions $u$ and $\varphi$ (again we are omitting the index $i$). 
Here it will be more convenient to work with equation 
(\ref{eq_u_g}). Divide (\ref{eq_u_g}) by $u$, integrate with 
respect to $dV_g$ and integrate by parts the Laplacian term to get
\begin{gather}
\int_M \frac{|\nabla_g u|^2}{u^2}\, dV_g  + \int_M \big (-\frac{1}{3}R_g + \frac{1}{6} F_g - \frac{1}{6} \Tsg \big ) \, dV_g
 = - \int_M \frac{1}{u} dV_g .
\label{basic_equality_g}
\end{gather}
By (\ref{transformation_orientifold}), the integral 
\begin{gather}
 \int_M \Tsg \, dV_g 
\nonumber
\end{gather}
is bounded by some constant $C>0$ depending only on the fixed quantities $\Ts$ and $g_0$. Hence, dropping out the non-negative terms,
\begin{gather}
 \int_M \frac{1}{u} dV_g \leq 
\int_M \frac{1}{3}R_g \, dV_g +  \int_M \Tsg \, dV_g \leq C, 
\label{1_over_u_beta_n_2}
\end{gather}
where the hypothesis that the integral of the scalar curvature is bounded has
been used.
As before, applying Jensen's inequality for the function $f(x) = \frac{1}{x}$, $x>0$, and 
the measure $dV_g$ leads to 
\begin{gather}
 \frac{1}{\dashint_M u \, dV_g} \leq \dashint_M \frac{1}{u} dV_g ,
\nonumber
\end{gather}
or equivalently
\begin{gather}
 \frac{1}{\int_M u \, dV_g} \leq \frac{1}{A^2} \int_M \frac{1}{u} dV_g ,
\nonumber
\end{gather}
where we used $\int_M dV_g = \int_M e^{n\varphi} \, dV_{g_0}= A$. 
In light of (\ref{contradiction_beta_n_2}), the right hand side of the above inequality
goes to infinity, contradicting (\ref{1_over_u_beta_n_2}).
This proves (\ref{bound_1}) when $\beta = \frac{n}{2}$
or (\ref{transformation_orientifold}) holds.

The statement (\ref{bound_dim_2}) now follows from the the Gauss-Bonnet 
formula as $\int_M R_g \, dV_g$ is a topological invariant, and hence bounded. \\

\noindent \emph{Proof of theorem \ref{main_theorem}-(3)}:
Use again (\ref{basic_equality_g}). It then follows from the hypothesis that
\begin{gather}
\int_M \frac{1}{u} \, dV_g \leq C,
\nonumber 
\end{gather}
which again gives a contradiction after an application of the Jensen's inequality.

\hfill $\qed$ \\

\noindent \emph{Proof of proposition \ref{little_prop}}: 
If the result is not true then from the proof of theorem \ref{main_theorem} we have
$\int_M e^{s \varphi}\,dV_{g_0} \rar 0$ for $1 \leq s < n$.
Since the exponential is a convex function,
we can use Jensen's inequality with the measure $dV_{g_0}$ to obtain:
\begin{gather}
\int_M e^{s \varphi}\,dV_{g_0} = \vol_{g_0}(M) \dashint_M e^{s\varphi}\,dV_{g_0}
\geq 
\vol_{g_0}(M) e^{s \frac{1}{\vol_{g_0}(M)} \int_M \varphi \,dV_{g_0} } ,
\nonumber
\end{gather}
where $\dashint_M =\frac{1}{\vol_{g_0}} \int_M$. So $\int_M e^{s \varphi} \rar 0$ implies 
$\int \varphi \, dV_{g_0} \rar -\infty$,  
and therefore
\begin{gather}
\infty  \leftarrow  \Big | \int \varphi \, dV_{g_0} \Big | \leq 
\int | \varphi | \, dV_{g_0}  = \,
\parallel \varphi \parallel_{L^1(M,g_0)}  ,
\nonumber
\end{gather}
contradicting $\varphi \in \cS \cap B_\eta(0)$. \hfill $\qed$\\

\subsection{Remarks}
\label{proof_remarks}

The inequality Eq. (\ref{basic_inequality}) is also interesting as a limit on how negatively curved the
compactification manifold $M$ can be.  In \cite{Douglas:2010rt} it was pointed out that $M$ cannot
have negative curvature unless $T_{st}>0$, but this was not quantified.  From equation (\ref{basic_inequality})
one has (see also proposition \ref{existence_2})
\begin{gather}
R_{g_0} \ge - \frac{\int_M e^{2(1-\beta) \varphi} \Ts \, dV_{g_0}}{2\vol_{g_0}(M)} .
\nonumber
\end{gather}
Thus the negative curvature can be no greater than string scale, and much less if the orientifolds live
on submanifolds.

We mentioned in section \ref{results} that if the effective potential is not bounded from below,
then the singularities that may form are generally of a very particular type. To see this, 
remember that we showed that in this setting we have
\begin{gather}
 \int_M e^{s \varphi_i } \, dV_{g_0} \rar 0,
\label{s_less}
\end{gather}
when $i \rar \infty$, for any $1\leq s < n$. Since $\int_M e^{s \varphi_i } \, dV_{g_0} = A$, 
for all $i$, a simple
application of the interpolation inequality (see appendix) yields
\begin{gather}
\int_M e^{s \varphi_i } \, dV_{g_0} \rar \infty
\label{s_greater}
\end{gather}
for any $ s > n$. On the other hand, (\ref{s_less}) gives that, up to a subsequence,
$e^\varphi_i \rar 0$ almost everywhere, and by Egoroff's theorem (see e.g. \cite{Rudin}) the convergence is uniform outside a set
of measure $\de$, where $\de > 0$ is as small as desired. This implies that $e^{n \varphi_i}$
has a behavior very much like a (sum of) Dirac delta(s): it converges to zero in most of $M$, 
blows-up in some localized subsets and has constant integral
(since $\int_M e^{n \varphi_i} d V_{g_0} = A$).

\section{Existence of critical points, topological considerations and some a priori estimates \label{existence_critical}}

In this section we discuss the existence of critical points and its interplay with the 
topology/geometry of the manifold $(M,g_0)$. First we prove existence theorems 
for equation (\ref{eq_u}) which also allow us 
to provide useful bounds for the critical points and the map $\mathfrak{F}_{g_0,\Ts,\cF}$.
We then discuss the general solvability
of equation (\ref{eq_v_al}), including non-positive values of $\al$.

Recall from \S 2 that our problem is to solve equation (\ref{eq_v_al}),
\begin{gather}
 P_g v = \Delta_g v + (-\frac{1}{3} R_g + \frac{1}{6} F_g - \frac{1}{6} \Tsg ) v =  \frac{1}{6} \al
\nonumber
\end{gather}
for some $\alpha\in\RR$ and an everywhere positive function $v$ satisfying equation (\ref{const_G_N}),
\begin{gather}
\int_M v\, dV_g = \frac{1}{G_N} .
\nonumber
\end{gather}
Of course, if we can solve (\ref{eq_v_al}) for some $\alpha$, then by
simultaneously rescaling $\al$ and $v$ we can solve (\ref{const_G_N}).  Thus, for fixed data $(g,F,T_{st})$,
we can restrict attention to the cases $\alpha=\pm 1$ or $0$, as in equation (\ref{eq_u_g}).
On the other hand, if we are considering a family of solutions in which $\alpha$ changes sign, we should
keep $\alpha$ general.

The solvability of (\ref{eq_v_al}) is governed by the associated eigenvalue problem
\begin{gather}
P_g v = \lambda v.
\nonumber
\end{gather}
We recall its basic features, granting that the coefficient functions are smooth,
$\Delta_g$ is negative definite and that $P_g$ is self-adjoint (see e.g. \cite{Evans, GT}):
\begin{itemize}
\item The spectrum $\Spec P_g  \subset \RR$ is countable and has no accumulation points. 
\item It is bounded above, $\Spec P_g \subset (-\infty,\lambda_0]$.
\item The eigenvalue $\lambda_0$ is isolated, meaning $\Spec P_g = (-\infty,\lambda_1] \cup \{\lambda_0\}$
with $\lambda_1<\lambda_0$.
\item The eigenspace to $\lambda_0$ is one dimensional and spanned by an everywhere positive eigenfunction.
\item The inhomogeneous problem $(P_g-\mu) u=f$
has a unique solution for each $f \in C^\infty(M)$ if and only if $\mu \notin \Spec P_g$. 
Such a solution $u$ is smooth. 
\end{itemize}
Thus, if $\lambda_0<0$, $P_g$ will have no kernel, so that existence and uniqueness of $u$ is evident.  
On the other hand, if $\lambda_0\ge 0$ or if we cannot show $\lambda_0<0$ {\it a priori}, the discussion
will be more complicated.  In this case we will try to argue by continuation from the easier $\lambda_0<0$ case.

Let us first consider cases where we can state conditions
which guarantee $\lambda_0<0$, which we do in propositions \ref{existence_1} and \ref{existence_2}.
The idea behind their proofs is 
to guarantee the existence of a maximum principle, which can then be combined with the  Fredholm alternative 
to reduce the problem of existence of solutions to that of 
uniqueness\footnote{It is only the sign of the lower order term which 
plays a crucial role in this procedure, but rather than simply assuming such a sign condition,
we provide separate bounds in terms
of the physical quantities composing the lower order coefficient, namely, the 
scalar curvature, the string term and the gauge fields. Since it is only the combination of these terms which is relevant
for the maximum principle argument, the choice of such bounds involves a great deal of arbitrariness, but
the reader can easily adapt the proof to other situations. We also seek to write conditions in a very simple fashion
(we provide a bound essentially in terms of the dimension), so that one can hope to actually verify them in concrete 
situations.} (see the appendix for details).

For simplicity, let us assume that $\beta$ is an integer and that the gauge fields $F^{(p)}$ are 
labeled according to their degree, which we still denote by $p$, so that $p=1,\dots,n$.
The proofs can easily be extended to other cases.

\begin{prop}
Assume the same hypotheses of theorem \ref{main_theorem}, and
suppose further that
$(M,g_0)$ has positive Yamabe invariant. Define
\begin{gather}
 K_{g_0,n}(\ve) = \frac{R_{g_0} - 3 \ve}{3n^n} .
\nonumber
\end{gather}
Consider the conditions: \\
\indent  (i) $\parallel \Ts \parallel_{C^0(M)} < K_{g_0,n}(\ve)$.\\
\indent (ii) $\parallel |F^{(p)}|^2_{g_0} \parallel_{C^0(M)} < K_{g_0,n}(\ve),~p=1,\dots,n$. \\
\indent (iii) $\parallel \Delta_{g_0} \varphi \parallel_{C^0(M)} < K_{g_0,n}(\ve)$. \\
\indent (iv) $\parallel |\nabla_{g_0} \varphi |^2_{g_0} \parallel_{C^0(M)} \leq K_{g_0,n}(\ve)$.\\
\indent (v) $\frac{1}{n} \leq e^{2\varphi(x)} \leq n^n$ for all $x \in M$. \\
\noindent Given $0 \leq \ve  < \frac{1}{3} R_{g_0}$, assume conditions (i)-(v). Then
equation (\ref{eq_u}) has a unique solution $u$. Such solution is smooth and positive. Moreover, if $\ve > 0$, 
this solution obeys the estimate
\begin{gather}
 \parallel u \parallel_{C^0(M)} \leq \frac{n^n}{\ve} .
\label{C_0_estimate_u}
\end{gather}
Also, if $\ve=0$ then the map $\mathfrak{F}_{g_0,\Ts,\cF}$ satisfies
\begin{gather}
\mathfrak{F}_{g_0,\Ts,\cF}(\varphi) > - \frac{ 3n(n+3)R_{g_0} }{ 2 A G_N^2  }
\label{lower_bound_F}
\end{gather}
where $K_{g_0,n} =  K_{g_0,n}(0)$.
\label{existence_1}
\end{prop}
In order to appreciate the relevance of proposition \ref{existence_1}, let us consider 
one of the cases of primary interest, namely,  $D=10$, so that $n=6$, and $K_{g_0,6} \sim 10^{-5} R_{g_0}$. 
Here all of $R_{g_0}$, $\Ts$ and $\cF$ have units
of $1/\Length^2$, or $\Mass^2$  in Planck units $\hbar=c=1$.
Assume for concreteness that $M$ is roughly isotropic in the sense that $\vol_{g_0}M \sim (R_{g_0})^{-6/2}$.
Then proposition \ref{existence_1} 
guarantees existence and uniqueness of a critical point for any
value of of $\Ts$ and $\cF$ up to the order of 
$10^{-5} (M^{(10)}_P)^2$.

The known applicability of this proposition is to compactifications on manifolds with Einstein
metrics of positive scalar curvature, such as the sphere.  These lead to $\al<0$ and compactifications
to anti-de Sitter space-time.  However it is interesting to note that
when $n=6$ and $M$ is a Calabi-Yau manifold, it is known that its 
topological Yamabe invariant 
is positive (proposition $4.3$ in \cite{Lebru}). Hence, in such cases, there exists at least 
one conformal class whose Yamabe invariant is positive and therefore 
there are solutions to equation (\ref{eq_u}) with positive scalar curvature, very different
from the Ricci flat metrics usually considered.  A physics discussion of this is given in \cite{Hertog:2003ru}.

One interesting consequence of (\ref{C_0_estimate_u}) is to give a control of critical points in terms 
of the (constant) scalar curvature $R_{g_0}$. In fact, choosing $\ve = \frac{1}{4}R_{g_0} < \frac{1}{3} R_{g_0}$, 
(\ref{C_0_estimate_u}) gives
\begin{gather}
 \parallel u \parallel_{C^0(M)} \leq \frac{4n^n}{R_{g_0}},
\label{u_bound_scalar_curvature}
\end{gather}
and therefore $u$ has to be very small if the scalar curvature is 
very large. This behavior of $u$ has already been 
identified in \cite{Do} by an argument based on the AdS/CFT correspondence, and
(\ref{u_bound_scalar_curvature}) provides a further refinement of that prediction. \\

\noindent \emph{Proof of proposition \ref{existence_1}}: Write the operator $M_{g_0}$ as
\begin{gather}
 M_{g_0} u =  \Delta_{g_0} u + (n-2) \langle \nabla_{g_0} \varphi, \nabla_{g_0} u \rangle_{g_0} 
-V u
\label{eq_u_V} 
\end{gather}
where $V= -\frac{2}{3}(n-1) \Delta_{g_0} \varphi - \frac{1}{3}(n-1)(n-2) |\nabla_{g_0} \varphi|^2_{g_0} 
+ \frac{1}{3}R_{g_0} - \cF + \frac{1}{6} e^{2(1-\beta)\varphi}\Ts$. We will show that conditions $(i)-(v)$ imply
\begin{gather}
 \inf_M V(x) > \ve.
\label{V_greater}
\end{gather}
We have 
\begin{gather}
\begin{split}
V > \frac{1}{3}R_{g_0} -\frac{2}{3}(n-1) K_{g_0,n}(\ve) -\frac{1}{3}(n-1)(n-2) K_{g_0,n}(\ve) \label{V_greater_K} \\
- \frac{1}{6}  \sum_{p=1}^L e^{2(1-p)\varphi} K_{g_0,n}(\ve) - \frac{1}{6} e^{2(1-\beta)\varphi}K_{g_0,n}(\ve).
\end{split}
\end{gather}
Since we are assuming $\beta$ to be an integer, we can write the last two terms into a single sum which is greater 
than or equal to 
\begin{gather}
\begin{split}
- \frac{2}{6} K_{g_0,n}(\ve) \sum_{\ell=0}^n \Big( e^{2\varphi}  \Big )^{1-\ell}
\geq - \frac{1}{3} K_{g_0,n}(\ve) \Big ( 1 + e^{2\varphi} 
+ \sum_{\ell=2}^n \Big( \frac{1}{e^{2\varphi}}  \Big )^{\ell-1} \Big ) \label{sum_greater_K} \\
\geq  - \frac{1}{3} K_{g_0,n}(\ve)\Big ( 1 + n^n + (n-2)n^{n-1} \Big ) \geq -\frac{2}{3} n^nK_{g_0,n}(\ve) , 
\end{split}
\end{gather}
where in the next to the last step we used $(v)$. Using (\ref{sum_greater_K}) in (\ref{V_greater_K}) gives
\begin{gather}
V > \frac{1}{3}R_{g_0} -\frac{1}{3}n(n-1) K_{g_0,n}(\ve) -\frac{1}{3} n^nK_{g_0,n}(\ve)  \nonumber \\
\geq \frac{1}{3}R_{g_0} -\frac{1}{3}n^n K_{g_0,n}(\ve) -\frac{2}{3} n^n K_{g_0,n}(\ve)  = 
\frac{1}{3}R_{g_0} -n^n K_{g_0,n}(\ve) . \nonumber
\end{gather}
The right hand side is equal to $\ve$ by the definition of $K_{g_0,n}(\ve) $,
and so (\ref{V_greater}) follows. In particular $V > 0$, including when $\ve=0$ 
(we get a strict inequality by the compactness of $M$), 
and it then follows that $M_{g_0}$ satisfies the maximum
principle (see appendix), and therefore $M_{g_0}$ has trivial kernel. We give the argument for completeness.
If $M_{g_0} w = 0$ then, from $M_{g_0} w \geq 0$ we get by the maximum principle that $w$ cannot have a 
non-negative maximum, hence $w\leq 0$. Analogously $M_{g_0} w \leq 0$ implies that $w$ cannot
have a non-positive minimum, hence $w \geq 0$ and therefore $w \equiv 0$. Now by the Fredholm alternative,
(\ref{eq_u_V}) has a 
unique solution $u$, and this solution is is smooth. To see that $u>0$, we again evoke 
the maximum principle. We have $M_{g_0} u = -e^{2 \varphi} \leq 0$.
Hence $u$ cannot have a non-positive minimum.

To prove (\ref{C_0_estimate_u}), define $\widehat{u} = u - \frac{n^n }{\ve}$. Using
(\ref{eq_u_V}),
\begin{gather}
 M_{g_0} \widehat{u} = \Delta_{g_0} \widehat{u} + (n-2) \langle \nabla_{g_0} \varphi, \nabla_{g_0} \widehat{u} \rangle_{g_0} 
- V \widehat{u} = -e^{2\varphi} + \frac{n^n }{\ve} V > -e^{2\varphi} + n^n \geq 0, 
\nonumber
\end{gather}
where (\ref{V_greater}) and $(v)$ have been used.
Evoking the maximum principle once more we obtain $\widehat{u} \leq 0$, and 
(\ref{C_0_estimate_u}) follows.

Next we show (\ref{lower_bound_F}). Let $x_0 \in M$ be a point where $u$ attains its minimum.
Then $\Delta_{g_0} u(x_0) \geq 0$ and $\nabla_{g_0} u(x_0) = 0$. Therefore 
\begin{gather}
 -e^{2\varphi(x_0)} = M_{g_0} u(x_0) \geq -V(x_0) u(x_0) ,
\nonumber
\end{gather}
and so
\begin{gather}
 u(x_0) \geq \frac{e^{2 \varphi(x_0)} }{V(x_0)} \geq \frac{1}{n\sup_M V}
\nonumber
\end{gather}
But
\begin{gather}
V \leq | V | \leq  \frac{2}{3}(n-1) |\Delta_{g_0} \varphi| + \frac{1}{3}(n-1)(n-2) |\nabla_{g_0} \varphi|^2_{g_0} 
+ \frac{1}{3}R_{g_0} + \cF + \frac{1}{6}e^{2(1-\beta)\varphi} |\Ts| \nonumber \\
< \frac{1}{3}n(n-1) K_{g_0,n} + \frac{1}{3} R_{g_0} 
+ \frac{2}{6} K_{g_0,n}\sum_{\ell=0}^n \Big( e^{2\varphi}  \Big )^{1-\ell} \nonumber \\
\leq \frac{1}{3}n(n-1) K_{g_0,n} + \frac{1}{3} R_{g_0} +
\frac{1}{3} K_{g_0,n}(n+1) n^n \leq (n+3)R_{g_0} \nonumber 
\end{gather}
where the definition of $K_{n,g_0}$ has been used. Hence 
\begin{gather}
 u(x_0) \geq \frac{ 1}{ n (n+3) R_{g_0} }.
\nonumber
\end{gather}
Now it follows that
\begin{gather}
\mathfrak{F}_{g_0,\Ts,\cF}(\varphi) = - \frac{6}{4G_N^2 \int_M e^{n\varphi} u \, dV_{g_0} }
\geq - \frac{6}{4G_N^2 \, \inf_M u \, \int_M e^{n\varphi}  \, dV_{g_0} } \geq
- \frac{ 3n(n+3)R_{g_0} }{ 2A G_N^2  } ,
\nonumber
\end{gather}
where we used that $\int_M e^{n\varphi}  \, dV_{g_0} = A$. \hfill $\qed$

\subsection{Non-positive Yamabe invariant}

\begin{prop}
Assume the same hypotheses of theorem \ref{main_theorem}, and
suppose further that
$(M,g_0)$ has non-positive Yamabe invariant. Let
 $H = \frac{1}{6n^{\vert 1-\beta \vert}}(-n^n\sum \vert F^{(p)} \vert _{g_0} ^2 +T_{st} ^{g_0})$ . Given $\ve > 0$, and $\Ga > 1$, 
define
\begin{gather}
 K_{g_0,n}(\ve,\Ga) = \frac{n^2\Ga}{3} + \frac{|R_{g_0}|}{3} + \ve  .
\nonumber
\end{gather}
Assume that 
\begin{gather}
\parallel \Delta _{g_0} \varphi \parallel_{C^0(M)} < \Ga \nonumber   \\
\parallel |\nabla _{g_0} \varphi|_{g_0}^2 \parallel_{C_0(M)} < \Ga \nonumber   \\
\frac{1}{n} \leq e^{2\varphi} \leq n \nonumber  \\
H > K_{g_0,n}(\ve,\Ga) \nonumber
\end{gather}

Then, there exists a unique, smooth and, positive solution to (\ref{eq_u}). This solution satisfies 
\begin{gather}
\parallel u \parallel _{C^0(M)} \leq \frac{n}{\epsilon}  .
\label{ubound} 
\end{gather}
Moreover, the map $\mathfrak{F}_{g_0,\Ts,\cF}$ satisfies, 
\begin{gather}
\mathfrak{F}_{g_0,\Ts,\cF} > -\frac{n^{1+\vert 1-\beta \vert} \Vert T_{st} ^{g_0} - \frac{1}{n^n} \sum_{p=1}^n \vert F^{(p)} \vert^2    \Vert _{C^{0}(M)}}{2G_N ^2 A} .
\label{funcbound}
\end{gather}
\label{existence_2}
\end{prop}

It is known that there are no topological obstructions to negative scalar curvature \cite{Au3}. In other 
words, on any compact manifold we can find a metric $g$ such that the total scalar curvature functional is negative, i.e.,
$\int_M R_g \,dV_g <0$, so $(M,g)$ will have negative Yamabe invariant.  Thus we have conditions which guarantee
a solution in this case; however they may require large $\Ts$ which is somewhat unphysical.  Nevertheless this could
be useful as a starting point for the analysis of families of solutions below.

For the situation of zero Yamabe invariant, we notice that once more in the case of primary interest, namely,
when $M$ is a Calabi-Yau, we can find a metric fulfilling the hypotheses of proposition \ref{existence_2}.
Simply choose a Ricci-flat metric on $M$.
Therefore, as in proposition \ref{existence_1}, our hypotheses in proposition 
\ref{existence_2} are not vacuous. \\

\noindent \emph{Proof of proposition \ref{existence_2}}: 
The proof is similar to proposition \ref{existence_1}, so we will only sketch the arguments.
Just as in the proof of proposition \ref{existence_1}, we may easily show that 
the hypotheses imply (\ref{V_greater}). Then, the same maximum principle type of argument as before 
shows there exists a unique, smooth and positive solution to (\ref{eq_u}).

To show (\ref{ubound}), if $x_0$ is a point of maximum of $u$, 
then, $\nabla _{g_0} u (x_0) =0$ and $\Delta _{g_0} u (x_0) \leq 0$. Hence, 
\begin{gather}
V(x_0) u(x_0) \leq e^{2 \varphi (x_0)}. \nonumber
\end{gather}
So, $u \leq \frac{n}{\epsilon}$.

Finally, we prove (\ref{funcbound}). At a point of 
minimum $y_0$ of $u$, $\nabla _{g_0} u (y_0) =0$ and $\Delta _{g_0} u (y_0) \geq 0$. Hence, 
\begin{gather}
V(y_0) u(y_0) \geq e^{2 \varphi (y_0)} \nonumber 
\end{gather}
But, 
\begin{eqnarray}
V &<& \frac{1}{3} |R_{g_0}| + \frac{1}{3} n(n-1)\Ga + \sup \frac{e^{2(1-\beta)\phi}}{6}\vert T_{st} ^{g_0} - 
\sum_{p=1}^n \vert F^{(p)} \vert^2 e^{2(\beta-p)\phi}  \vert \nonumber \\
&<& \frac{n^{\vert 1-\beta \vert}}{3} \Vert T_{st} ^{g_0} - \frac{1}{n^n} \sum_{p=1}^n \vert F^{(p)} \vert^2    \Vert _{C^{0}(M)}, \nonumber   
\end{eqnarray}
and the result follows. \hfill $\qed$ \\

\subsection{Families of solutions}
The basic point we want to establish is
\begin{conjecture}
Suppose we have fixed $(M,g_0,\Ts)$ and we are given data $(\varphi,\cF)$ depending smoothly on
a real parameter $t$, such that for $t=0$ we are in one of the situations governed by the preceding
propositions and thus $\al|_{t=0}<0$.  Then, there will be a region $t\in [0,t_1]$ for some $t_1>0$
in which equation (\ref{eq_v_al}) has a unique solution satisfying (\ref{const_G_N}),
and in some cases $\al|_{t=t_1}>0$.
\end{conjecture}
In particular, we might add a constant term $C t$ to $\cF$, causing the eigenvalues
of $P_g$ to shift as $\lambda \rightarrow \lambda + C t$; then for a given $t_1$ and sufficiently large $C$ 
we will have $\al|_{t=t_1}>0$.  Physically, this connects a compactification to anti-de Sitter space-time
to another compactification to de Sitter space-time, and is referred to as ``uplifting'' \cite{Douglas:2006es,KKLT}.

Let us give some intuition of why one expects the above conjecture, or some variation of it, to be true.
Begin by assuming that the corresponding solutions
$v = v_\al$ vary in a well behaved manner (say, they depend continuously or smoothly on the data). 
What can be said about the solutions $v_0$ corresponding to
$\al=0$?

The (unique)
solution of (\ref{eq_v_al}) can be written as an eigenfunction expansion of the form
\begin{gather}
 v = \frac{\al}{6} \sum_i \frac{1}{\la_i} \left( \int_M \psi_i \, dV_g \right) \psi_i,
\label{v_expansion}
\end{gather}
and the effective potential is then given by
\begin{gather}
 \frac{1}{\cV} = \frac{2}{3} G_N^2 \sum_i \frac{1}{\la_i} \left( \int_M \psi_i \, dV_g \right)^2.
\label{V_eff_expansion}
\end{gather}
Equality (\ref{V_eff_expansion}) indicates the presence of a resonant mode when 
$\cV$ approaches zero. In other words, when $\al \rar 0$, one has $\la_s \rar 0$ for one of the eigenvalues $\la_s$ 
in $\{ \la_i \}$; notice that by (\ref{v_expansion}), if $\al \rar 0$ then having $\la_s \rar 0$ is the 
only way we can have a non-trivial $v$ (at least under the present assumption that $v$ will vary smoothly
with the data of the problem).
In this case, the main contribution to the warp factor will come from the corresponding 
eigenfunctions $\psi_s$, as can be seen from (\ref{v_expansion}). 
It is also worth noticing that if $\la_s = 0$ then $\psi_s$ will solve
\begin{gather}
P_g \psi_s = 0,
\nonumber
\end{gather}
which is exactly  (\ref{eq_v_al}) with $\al = 0$. 

In order to be an acceptable solution, one needs $v$ to be positive. If $\psi_s$ is the ground 
state, i.e., $\la_s = 0$ is the principal eigenvalue, then the corresponding 
eigenspace is one-dimensional and we can assume $\psi_s >0$, and hence $v = \psi$.
But if the resonant state is an excited one, then there might be more than one corresponding $\psi_s$, 
and generally these will take positive and negative values. Since the dominant contribution to 
$v$ will be from the $\psi_s$'s, generally $v$ will be zero or negative 
somewhere on $M$. Hence we seek to investigate whether this situation can be avoided, at least 
in some generic sense.

Suppose $\cV$ crosses from negative to positive values, and consider $\cV$ right before
it becomes zero, so $\cV \approx 0$ but $\cV < 0$. Then, up to a positive constant, (\ref{V_eff_expansion}) 
becomes
\begin{gather}
 \frac{1}{\cV} \approx \frac{1}{\la_s},
\nonumber
\end{gather}
where $\la_s$ corresponds to the energy level that becomes resonant when $\cV = 0$. 
The condition $\cV < 0$ gives $\la_s < 0$, and therefore (\ref{v_expansion}) implies
\begin{gather}
 v \approx \frac{\al}{\la_s} \sum_{j=1}^k \psi_{s_k},
\label{v_apporx}
\end{gather}
where $\psi{s_k}$ are the eigenfunctions associated to $\la_s$; notice that 
$\frac{\al}{\la_s} > 0$ since $\cV < 0 \Leftrightarrow \al < 0$.

Now recall that we are assuming that for $\al < 0$, or equivalently $\cV < 0$, (\ref{eq_v_al}) has 
a (unique) solution $v$, which satisfies $v >0$. As we pointed out earlier, typically
one expects the eigenfunctions associated to excited states to take positive and negative values,
and therefore the positivity of $v$ along with (\ref{v_apporx}) suggest that we are in the case in which
$\la_s$ is the principal eigenvalue, so that the sum on the right hand side of (\ref{v_apporx}) contains
only one term $\psi_s$ and it satisfies $\psi_s > 0$. Therefore, under our assumptions, we expect
that when $\cV =0$, the resonance $\la_s = 0$ will be the principal eigenvalue, and in this case the warp factor will
be $v = \psi_s > 0$ (up to multiplication to positive constants).

Summarizing, the above arguments suggest that typically one can continue $\cV$ through zero. Other
arguments in this direction have already been given in \cite{Do}. Moreover, the study of continuous dependence
of eigenvalues of the Laplacian \cite{BU} and analytic properties of the resolvent \cite{Kato} suggest that
a picture where some sort of smooth or continuous dependence, as assumed above, is likely to hold.

Finally, it is worth pointing out that if we rewrite all quantities in terms of a fixed constant scalar curvature 
background metric $g_0$, then 
our basic equation becomes
\begin{gather}
 L_{g_0} u - \frac{1}{3} R_{g_0} \, u = \frac{1}{6}  e^{2\varphi} \al
\nonumber
\end{gather}
and the study of eigenvalues previously discussed can be rephrased into the
more geometric question of whether a multiple of the scalar curvature, $\frac{1}{3} R_{g_0}$, 
is an eigenvalue of the linear operator $L_{g_0}$.

\section{The case $d \neq 4$ \label{d_not_4}}
Although $d=4$ for compactifications of string theory which could describe fundamental physics, there is also a good deal of 
physics work on $d \neq 4$.  Such compactifications can be much simpler
than $d=4$, for example the compactification of the IIb superstring on $S^5$ leading to
anti-de Sitter space-time \cite{Aharony}.  A compactification with $d>4$ can also be further compactified
to get $d=4$.

The reasoning of section \ref{results}, when applied to the general $d$ case yields the following 
equation  for the critical points of 
the action (see \cite{Do} and equation $(2.33)$ therein for details)
\begin{gather}
\Delta_g v -\frac{1}{2(d-1)} \Big (\frac{d}{2}R_g + T^{(d)} \Big )v = \frac{(d-2)\al}{4(d-1)}v^{1-\frac{4}{d}} ,~~v>0, 
\label{general_d}
\end{gather}
where $T^{(d)} = -\frac{d}{2} F_g + \Tsg = -\frac{d}{2}\sum|F_g^{(p)}|^2 + \Tsg$ and $\al$ is the 
Lagrange multiplier in (\ref{effective_potential}). 
One sees that the case $d<4$ is qualitatively different because of the negative exponent in the source term,
and we will not treat it here.\footnote{
The case $d=2$ is substantially different as one should use different constraints than those of Conjecture \ref{main_conjecture},
{\it i.e.}  one should not go to Einstein frame.} 

Recalling that the role of $\al$ is to impose the constraint 
\begin{gather}
 \int_M v^{2-\frac{4}{d}} \, dV_g = \frac{1}{G_N},
\label{constraint_general_d}
\end{gather}
we see that it suffices then to find  $v > 0$ solving the problem
\begin{gather}
\Delta_g v +  f_g v = K v^{1-\frac{4}{d}}, 
\label{general_d_f_K}
\end{gather}
where $f_g$ is a a given function, possibly depending on $g$,
and $K$ is a non-zero constant.
In fact, if $v$ solves (\ref{general_d_f_K}) with
$f = -\frac{1}{2(d-1)} (\frac{d}{2}R_g + T^{(d)} )$, then
$\widetilde{v} = a v$ satisfies 
(\ref{constraint_general_d}) and solves the equation with $K$ replaced 
by $K a^\frac{4}{d}$, where 
\begin{gather}
a = \left( \frac{1}{G_N \int_M v^{2-\frac{4}{d}}\, dV_g} \right)^\frac{d}{2d-4}.
\nonumber
\end{gather}
The important aspect of $K$ is that it has the 
same sign of $\al$.
The case $K=0$, or $\al = 0$,
will be excluded here because in this case the equation is linear.

Let $L_g = \Delta_g + f_g$. Whenever a solution $v>0$ of (\ref{general_d_f_K}) exists,
the sign of $K$ is opposite to that of the lowest eigenvalue of $L_g$:
let $\la_1$ be the lowest eigenvalue with corresponding eigenfunction $\phi_1 > 0$, and
take the $L^2$ inner product to find, after integration by parts,
\begin{gather}
 \int_M \phi_1 L_g v \, dV_g = \int_M v L_g \phi_1 \, dV_g 
= -\la_1 \int_M \phi_1 v \, dV_g = K \int_M \phi_1 v^{1-\frac{4}{d}} \, dV_g.
\nonumber
\end{gather}
Hence, the sign of $K$, and therefore that of $\al$, can be determined by studying the eigenvalue problem
for the linear portion of the equation.\footnote{In particular, if $T^{(d)}$ is somehow tuned so that 
$f_g = -\frac{n-2}{4(n-1)} R_g$, then $L_g$ becomes the conformal Laplacian of the manifold $(M,g)$, and
therefore the sign of $\la_1$, and hence of $K$, is an invariant of 
the conformal class \cite{LP, Au2, SY1}.} Not surprisingly, then, 
different behaviors for solutions can be expected according to the sign 
of the first eigenvalue of $L_g$. With that in mind, we now turn to some
general existence results.

\begin{prop}
Assume $f_g$ is smooth, $K > 0$ and $d > 4$. Then problem (\ref{general_d_f_K})
has a smooth
non-negative solution $v$ satisfying
\begin{gather}
\parallel v \parallel_{C^0(M)} \, \leq \left( \frac{ K }{\parallel f_g \parallel_{C^0(M)}} \right)^\frac{d}{4}
\label{bound_v_K_pos} 
\end{gather}
\label{existence_general_d_1}
\end{prop}
\begin{proof}
Equation (\ref{general_d_f_K}) can be solved by a standard sub- and super-solutions method
(see proposition \ref{sub_super_sol}). It is enough to find two functions $v_- < v_+$ such that
\begin{gather}
 L_g v_- - K v_-^{1-\frac{4}{d}} \geq 0, \nonumber \\
 L_g v_+ - K v_+^{1-\frac{4}{d}} \leq 0. \nonumber
\end{gather}
For $v_-$ we can take $v_- \equiv 0$. If we choose $v_+$ to be a positive constant, then
\begin{gather}
 L_g v_+ - K v_+^{1-\frac{4}{d}} = f_g v_+ - K v_+^{1-\frac{4}{d}}
 \leq \, \parallel f_g \parallel_{C^0(M)} v_+ - K v_+^{1-\frac{4}{d}}.
\nonumber
\end{gather}
Setting the right hand side of the above expression to be less than or equal to zero implies
\begin{gather}
 v_+ \leq \left( \frac{K}{ \parallel f_g \parallel_{C^0(M)}} \right)^\frac{d}{4}.
\nonumber
\end{gather}
Therefore one can set
\begin{gather}
 v_+ \equiv \left( \frac{K}{ \parallel f_g \parallel_{C^0(M)}} \right)^\frac{d}{4}.
\nonumber
\end{gather}
Then, proposition \ref{sub_super_sol} guarantees the existence of a solution to (\ref{general_d_f_K})
satisfying $v_- \equiv 0 \leq v \leq v_+$. In particular
the bound (\ref{bound_v_K_pos}) holds. Smoothness then follows from elliptic regularity 
(proposition \ref{elliptic_regularity}).
\end{proof}
\begin{rema}
Without further information on $f_g$, we cannot guarantee that the solution $v$ found in 
\ref{existence_general_d_1} is strictly positive. For example, if $f_g$ vanishes identically
then $v \equiv 0$ is the only non-negative solution to (\ref{general_d_f_K}). 
\end{rema}

For $K<0$, we generally expect the negative part of $f_g$ to dominate. Indeed, since if $v\geq0$ then, integrating
(\ref{general_d_f_K}), gives
\begin{gather}
\int_M f_g v \, dV_g = K \int_M v^{1-\frac{4}{d}} \, dV_g \leq 0.
\label{sign_condition_f_g_K}
\end{gather}
Furthermore, if $f_g < 0$ and $v>0$ then, 
the integral on the left hand side of (\ref{sign_condition_f_g_K}) is negative, and
one sees that a necessary condition for the existence of non-negative solutions to (\ref{general_d_f_K})
is $K < 0$. This motivates the following:

\begin{prop}
Assume $f_g < 0$ is smooth, $K < 0$ and $d > 4$. Then problem (\ref{general_d_f_K})
has a smooth
positive solution $v$ satisfying
\begin{gather}
\parallel v \parallel_{C^0(M)} \leq \left(\frac{|K|}{\min_M |f_g|} \right)^\frac{d}{4}.
\label{bound_v_K_neg}
\end{gather}
\label{existence_general_d_2}
\end{prop}
\begin{rema}
Recall that, in the case of interest 
$f = -\frac{1}{2(d-1)} (\frac{d}{2}R_g + T^{(d)} )$. In this case 
the assumption on $f_g$ can be obtained from 
conditions on $R_{g_0}$, $\varphi$, $F_g$ and $\Ts$ in a similar fashion to what was done in the linear case 
(section \ref{existence_critical}), but we will not write them here for the sake of brevity.
\end{rema}
\begin{proof}
Since $K< 0$, we can write $K=-|K|$. Similarly to the proof of proposition \ref{existence_general_d_1},
we look for $v_- < v_+$ satisfying
\begin{gather}
 L_g v_- + |K| v_-^{1-\frac{4}{d}} \geq 0, \nonumber \\
 L_g v_+ + |K| v_+^{1-\frac{4}{d}} \leq 0. \nonumber
\end{gather}
Once again we take $v_- \equiv 0$. If $v_+$ is constant and positive,
\begin{gather}
 L_g v_+ + |K| v_+^{1-\frac{4}{d}} =
 f_g v_+ + |K| v_+^{1-\frac{4}{d}} \leq -\min_M |f_g| v_+ + |K| v_+^{1-\frac{4}{d}} ,
\nonumber
\end{gather}
where we used that $f_g < 0$ (notice that the minimum of $|f_g|$ will be non-zero by compactness of $M$).
Setting the right hand side of above expression to be less than or equal to zero implies
\begin{gather}
 v_+ \geq \left( \frac{|K|}{\min_M |f_g|} \right)^\frac{d}{4},
\nonumber
\end{gather}
and hence we can set
\begin{gather}
 v_+ \equiv \left( \frac{|K|}{\min_M |f_g|} \right)^\frac{d}{4}.
\nonumber
\end{gather}
From proposition \ref{sub_super_sol_prop} we obtain a non-negative solution $v$ obeying
the bound (\ref{bound_v_K_pos}). This solution is strictly positive. In fact, since
\begin{gather}
 L_g v = -|K|v^{1-\frac{4}{d}} \leq 0,
\end{gather}
and $f_g < 0$, the maximum principle guarantees that $v>0$ (unless $v$ is constant).
\end{proof}

To conclude this section, let us touch upon the boundedness of the effective potential
for the special case of $n=2$ and $d>4$, stating and proving a version 
of conjecture \ref{main_conjecture}. Indeed, by equation $(2.41)$ in \cite{Do}, we see that, 
the effective potential evaluated at a critical point is $\frac{(d-2)\alpha}{2dG_N}$. However, one 
cannot define a functional $\cF$ as we did earlier by simply plugging in a critical point. This is 
because, unlike the linear case $d=4$, if there is more than one critical point, $\cF$ may not be well-defined. 
This being said, in the case of $n=2$, it can be proven that, regardless of which critical point is plugged into 
the effective potential, it is still bounded from below. Firstly, we see that the statement is non-trivial
 only when $\alpha <0$. Hence, we shall assume that $\alpha = -\vert \alpha \vert$. Upon introducing a 
new function $u= v \vert \alpha \vert^{-\frac{d}{4}}>0$, equation (\ref{general_d}) becomes, 
\begin{gather}
\Delta _g u - \frac{1}{2(d-1)}\big (\frac{d}{2}R_g + T^{(d)}\big )\,u = -\frac{d-2}{4(d-1)} u^{1-\frac{4}{d}} 
\label{generaldnoa}
\end{gather} 
Imposing the warped volume constraint, it follows that
\begin{gather}
\cV(g) = \frac{d-2}{2d G_N} \frac{1}{\left(G_N \int_M u^{2-\frac{4}{d}} \, dV_g \right)^{\frac{2}{d-2}} }
\nonumber
\end{gather}

\begin{prop}
Let $d > 4$. Assume that there exists a constant $\eta > 0$ such that 
\begin{gather}
 \int_M T^{(d)} \, dV_{\widetilde{g}} \leq \eta,\,\, \text{ for all } \,\, \widetilde{g} \in [g].
\nonumber
\end{gather}
Then there exists a number $K_\eta \in \RR$ such that $\cV(g=e^{2\varphi}g_0) \geq K_\eta$ for all smooth
 functions $\varphi$ such that the volume of $g$ is fixed. 
\end{prop}
\begin{proof}
As usual we will use $C>0$ to denote several different constants independent of $\varphi$. 
Dividing equation (\ref{generaldnoa}) by $u$ and integrating (with respect to $dV_g$), we have,
\begin{gather}
-\int_M \frac{d-2}{4(d-1)} u^{-\frac{4}{d}}\, dV_g = \int_M \frac{\vert \nabla u \vert^2}{u^2} \, dV_g
- \frac{1}{2(d-1)}\Big (\frac{d}{2}\int_M R_g\, dV_g  + \int_M T^{(d)} \, dV_g \Big ) \nonumber  > -C. \nonumber
\end{gather}
where we have used that $\frac{d}{2}\int_M R_g\, dV_g $ is a topological invariant by the Gauss-Bonnet theorem, and
hence independent of $\varphi$.
Therefore,
\begin{gather}
 \int_M u^{-\frac{4}{d}} \, dV_g < C. \label{ineq}  
\end{gather}
Let $w=u^{\frac{4}{d}}$. Then,
\begin{gather}
\cV(g) = -\frac{C}{(\int_M w^{\frac{d}{2}-1} \,dV_g)^{\frac{2}{d-2}}} \geq
-\frac{C}{\int_M w \, dV_g} \geq
-C \int_M \frac{1}{w} \,dV_g  . \nonumber 
\end{gather}  
The last inequality (which is obtained by Jensen's inequality) along with inequality (\ref{ineq}) implies the result.  
\end{proof}

\section{Examples\label{examples}}

Here, it is pointed out how our results can be applied in some important examples. 
As we do not intend to exhaust all possible models, we will keep the discussion
short and somewhat informal, highlighting the main ideas 
and avoiding technicalities\footnote{As we will be looking at specific examples,
the metrics under consideration will be of the form $g= e^{2\varphi} g_0$ for some particular
family of conformal factors $\varphi$. It is clear that we can still apply theorem \ref{main_theorem}
when the infimum is taken over a smaller set $\mathcal{P} \subset \cS$.}.

As was stressed in section \ref{discussion}, important examples 
involving the string term are $\Ts = \La \de_{g_0}(p)$ and $\Ts = \La \de_{g_0}(N)$, for some $p \in M$, 
some function $\La$ which does not depend on the metric, and some submanifold $N$.
In both cases, (\ref{transformation_orientifold}) holds, and 
an inspection 
in the proof of theorem \ref{main_theorem}-(2) reveals that
this is essentially all that is needed for the proof to work. Indeed,
from (\ref{transformation_orientifold}) it follows that (\ref{1_over_u_beta_n_2}) remains valid, and the rest of 
of the arguments go through. We conclude that theorem (\ref{main_theorem}) can 
also be applied when the string term is a delta function or an orientifold plane (provided, of course, that the remaining 
hypotheses are still in place).

In a similar fashion, suppose $\Tsg$ takes a Gaussian shape, 
\begin{gather}
 \Tsg = \frac{1}{\sigma_g^n} \exp \left( {-\frac{r_g^2}{2 \sigma_g^2} } \right), 
\label{Gaussian_T}
\end{gather}
where $r_g$ is the Riemannian distance to a fixed point $p \in M$, 
and the ``width'' $\si_g$ is allowed to depend on $g$. For definiteness,
we can imagine to be working in the neighborhood of $p$, so that 
(\ref{Gaussian_T}) is multiplied by a suitably chosen cut-off function outside 
a ball centered at $p$.

In order to investigate how should $\si_g$ transform, one can
use the transformation law for a delta function,
\begin{gather}
 \de_g(p) = e^{-n\varphi} \de_{g_0}(p),~~\text{ for }~~ g = e^{2\varphi} g_0,
\label{transformation_delta}
\end{gather}
as a guide. Upon rescaling of the metric,
\begin{gather}
  g = \la^2 g_0
\nonumber
\end{gather}
one obtains,
\begin{gather}
\int_M \frac{1}{\sigma_g^n} \exp \left( -\frac{\la^2 \,r_{g_0}^2}{2 \sigma_g^2} \right) \la^n dV_{g_0} .
\nonumber
\end{gather}
If this is to equal 
\begin{gather}
\int_M \frac{1}{\sigma_{g_0}^n}  
\exp \left( -\frac{\,r_{g_0}^2}{2 \sigma_{g_0}^2} \right) \, dV_{g_0} ,
\nonumber
\end{gather}
then we see that $\si_g$ has to transform as
\begin{gather}
 \si_g = \la \si_{g_0}.
\nonumber
\end{gather}
Hence, requiring $\si_g$ to transform as above reproduces the behavior (\ref{transformation_delta}) 
for rescalings of the metric. 

For more general conformal transformations, $g=e^{2\varphi} g_0$, the distance function appearing in (\ref{Gaussian_T})
does not change in a simple way, and hence finding the correct transformation law 
for $\si_g$ is complicated. 
But when $\si_{g_0}$ is very small, we hope for a behavior similar to (\ref{transformation_delta}), 
since in this case $\Ts$ will approach a delta function. We can then 
impose on $\si_g$ a condition that guarantees a transformation law of the form
\begin{gather}
\frac{1}{\sigma_g^n} \exp \left( {-\frac{r_g^2}{2 \sigma_g^2} } \right)
= e^{-n \varphi} \frac{1}{\sigma_{g_0}^n} \exp \left( {-\frac{r_{g_0}^2}{2 \sigma_{g_0}^2} } \right) 
+ O(\si_0)~~ \text{ as } ~~\si_0 \rar 0.
\nonumber
\end{gather}
In this situation, the integral of the string term will again be bounded independent of $\varphi$,
and the arguments of theorem \ref{main_theorem} may be used to prove that the effective potential 
is bounded from below.

Another example involving radially symmetric functions is given by 
\begin{gather}
 e^{2\varphi} = \frac{1}{(a^2 + |x|^2)^\ga}
\label{gamma_factor}
\end{gather}
with $a > 0$ and $\ga > 1$.
This example has already been considered in \cite{Do}, so let us
investigate how it fits in the present results\footnote{In \cite{Do}, $\ga < 2$ is also 
imposed. This is to ensure that the scalar curvature is not negative for large $|x|$, but we will not 
need this condition here.}. Again for definiteness, a suitable 
cut-off function has be to used far away from the origin. More precisely, 
(\ref{gamma_factor}) is defined in $\RR^n$, with $|x|$ being the Euclidean distance to the origin,
and we are considering the conformally flat metric
\begin{gather}
 g_\ga = g = e^{2\varphi} \de = \frac{1}{(a^2 + |x|^2)^\ga} \de,
\label{g_gamma}
\end{gather}
with $\de$ being the Euclidean metric. Choosing an appropriate cut-off function for $|x|$ large, this metric 
can be glued to the manifold $(M, g_0)$, as indicated in \cite{Do}.
\begin{rema}
With the above gluing argument implicitly understood, we will work in $\RR^n$
for simplicity. In other words, we will
explore the behavior of $g$ in $\RR^n$, but our primary interest is in the restriction
of  all quantities to a 
fixed ball of large radius.
\label{remark_gluing}
\end{rema}
Firstly, notice that the constraint $\int_M e^{n\varphi} \, dV_g = A$ fixes $a$ in (\ref{gamma_factor}):
\begin{gather}
 \int_{\RR^n}\frac{dx}{(a^2 + |x|^2)^{\frac{n\ga}{2} } } 
= a^{-n(\ga-1)} \vol(S^{n-1}) \int_0^\infty \frac{s^{n-1}}{(1 + s^2)^{\frac{n\ga}{2} } } ds.
\nonumber
\end{gather}
The condition $\ga > 1$ ensures the convergence of the integral, so that
\begin{gather}
 a = a (\ga) = 
\left( \frac{  \vol(S^{n-1}) \int_0^\infty \frac{s^{n-1}}{(1 + s^2)^{\frac{n\ga}{2} } } ds }{ A } \right)^\frac{1}{n(\ga-1)}.
\label{a_gamma}
\end{gather}
Using (\ref{scalar_transformation}) one finds
\begin{gather}
 R_g(r) = (a^2 + r^2)^\ga \Big [ 2(n-1)r^{1-n} \Big ( \frac{ \ga n r^{n-1} }{ a^2 + r^2 } 
- \frac{ 2 \ga r^{n+1} }{ (a^2 + r^2)^2 } \Big )
- \frac{ \ga^2 ( n - 2 )(n-1) r^2 }{ (a^2 + r^2 )^2 } \Big ],
\label{scalar_radial}
\end{gather}
where $r = |x|$. It then follows that
\begin{gather}
 R(0) = 2(n-1)n\ga a^{2(\ga - 1)}.
\label{scalar_origin}
\end{gather}
From (\ref{a_gamma}), we see that $a \rar \infty$ when $\ga \rar 1$, and hence the scalar curvature
blows up at the origin. On the other hand, $a$ remains bounded 
when $\ga$ increases. Hence, if one fixes $\ga_0 > 1$ and considers
values of $\ga$ satisfying $\ga_0 < \ga$, the scalar curvature is bounded at the origin, and 
it is not difficult to see from (\ref{scalar_radial}) that it remains bounded on compact sets.
Therefore, given $\ga_0 > 1$ and $\rho >0$, there exists a constant $K_{\ga_0,\rho}$ such that
\begin{gather}
| R_{g_\ga} |(x) \leq K_{\ga_,\rho},
\nonumber
\end{gather}
for all $\ga > \ga_0$ and $|x| \leq \rho$, where $g_\ga$ is given by (\ref{g_gamma}).
Therefore theorem \ref{main_theorem} can be applied 
(after suitable adjustments as pointed out in remark \ref{remark_gluing}), and we
conclude that the effective potential is bounded from below for this family of metrics.

The previous example clearly resembles a well known phenomena on the round sphere,
which we now briefly recall (see \cite{LP} for details).

Identifying the complement of the north pole of $S^n$ with $\RR^n$ via stereographic 
projection, the round metric can then be written as
\begin{gather}
 g_0 = 4 u_0^\frac{4}{n-2} \de,
\nonumber
\end{gather}
where $u_0$ is the radially symmetric function
\begin{gather}
u_0(x) = (1+|x|^2)^\frac{2-n}{2},
\nonumber 
\end{gather}
which is usually referred to as the ``standard bubble''. Acting on $g_0$ with the group 
of conformal diffeomorphisms of the sphere generates the family of metrics
\begin{gather}
 g_\ve = 4 u_\ve^\frac{4}{n-2} \de, 
\nonumber
\end{gather}
where the one-parameter family of functions
\begin{gather}
u_\ve(x) = \ve^\frac{n-2}{2} (\ve^2+|x|^2)^\frac{2-n}{2},~~\ve > 0,
\nonumber 
\end{gather}
is also commonly known as the standard bubble. The presence of the parameter $\ve > 0$ is easily understood: 
the functions $u_\ve$, or equivalently the metrics $g_\ve$, come from dilations in $\RR^n$ inducing
conformal transformations on the sphere. All the metrics $g_\ve$ have the same volume, and
all of them are of constant scalar curvature equal to the round one. Hence, the volume constraint
and the bound on the scalar curvature required in theorem \ref{main_theorem} are satisfied 
for this family of metrics, and therefore the theorem can be applied. 
The important point here is that even though one restricts to constant scalar curvature metrics on the sphere,
there is an entire non-compact family of them, namely $g_\ve$, a situation where one would naturally
wonder if theorem \ref{main_theorem} holds.

The previous example admits a generalization in the following sense.
Suppose $n\geq 3$, fix a background metric $g_0$, and consider the set of $\varphi$ such that 
the metric $g = e^{2 \varphi} g_0$ has fixed volume and constant scalar curvature.
When the Yamabe invariant of $M$ is non-positive there exists a unique such $\varphi$
\cite{S2,S3,S4}. In the the positive case, however, it is possible to find a dense (in the $C^0$ topology)
subset $\mathcal{P}$ of such conformal factors satisfying the volume constraint and having 
constant scalar curvature \cite{Po}. As the scalar curvature remains unchanged for any $\varphi \in \mathcal{P}$,
the hypotheses of theorem \ref{main_theorem} are valid, and therefore the effective potential 
will be bounded from below when the infimum is taken over $\mathcal{P}$.

In the situation described in the previous paragraph, more can be said.
Due to compactness theorems for the Yamabe problem (\cite{KMS} and references therein),
the $C^{k,\al}$ norm of $e^{2 \varphi}$ will be uniformly bounded for all $\varphi \in \mathcal{P}$ 
by a constant depending only on the background 
metric $g_0$\footnote{Not to miss a subtle point, such theorems hold only when $n \leq 24$ and
the manifold is not conformally equivalent to the round sphere, and they also rely on the validity 
of the Positive Mass theorem in higher dimensions. But none of such issues arises in the cases of interest:
we have already discussed the case of the round sphere, and the Positive Mass theorem holds in dimensions
less than or equal to seven \cite{SY2, SY3}.}. Combining this with
equation (\ref{eq_u}) and Schauder estimates (see proposition \ref{schauder}) gives the following 
bound for the critical points:
\begin{gather}
 \parallel u \parallel_{C^{k,\al}(M)} \, \leq C_1 + C_2 \parallel u \parallel_{C^0(M)},
\nonumber
\end{gather}
where the constants $C_1$ and $C_2$ depend only the the background 
metric $g_0$ and $C^{0,\al}$-norms of $F_{g_0}$ and $\Ts$ (and of course on the 
dimension and on $\al$). Furthermore, if the solution $u$ is unique, then
the above estimate holds with $C_2 = 0$ (possibly after redefining $C_1$).
The important point here is that the growth of the derivatives of $u$ (and of 
$u$ itself when the solution is unique) is essentially controlled by 
the \emph{fixed} quantities $F_{g_0}$ and $\Ts$.

\appendix 

\section{Conventions and summary of formulae}

In this appendix we state some results we used in the proofs. Even though 
they are standard,
they are included here for the reader's convenience. 
The definitions and theorems here presented are far from general, but 
they will suffice to the purpose of this paper. We will comment on such generalizations, referring
 to the literature for details.
We will also indicate in this appendix our sign conventions.

\subsection{Inequalities}
Here we recall some standard inequalities. They can all be found, for example, in \cite{Evans}.\\

\noindent \textbf{Cauchy's inequality with $\ve$.}
\begin{gather}
 ab \leq \frac{\ve}{2} a^2 + \frac{1}{2\ve}b^2,~~a,b> 0,~\ve > 0
\nonumber
\end{gather}

\noindent \textbf{H\"older's inequality.} Assume $1 \leq p,q \leq \infty$, $\frac{1}{p} + \frac{1}{q} = 1$. 
Then if $u \in L^p(M)$, $v \in L^q(M)$, we have
\begin{gather}
\int_M |uv| \,dV \leq \parallel u \parallel_{L^p(M)} \parallel v \parallel_{L^q(M)}
\nonumber
\end{gather}

\noindent \textbf{Interpolation inequality.} Assume $1 \leq s \leq r \leq t \leq \infty$ and 
\begin{gather}
 \frac{1}{r} = \frac{\theta}{s} + \frac{1-\theta}{t} 
\nonumber
\end{gather}
Suppose that $ u \in L^s(M) \cap L^t(M)$. Then $u \in L^r(M)$ and 
\begin{gather}
 \parallel u \parallel_{L^r(M)} \leq \parallel u \parallel_{L^s(M)}^\theta \parallel u \parallel_{L^t(M)}^{1-\theta}
\nonumber
\end{gather}

\noindent \textbf{Jensen's inequality}. Assume $f: \RR \rar \RR$ is convex, $u: M \rar \RR$ is integrable
and $M$ has finite measure. Then
\begin{gather}
 f \Big (\dashint_M u \, dV \Big ) \leq \dashint_M f(u) \, dV 
\nonumber
\end{gather}
where $\dashint_M =\frac{1}{\vol(M)} \int_M$.

\subsection{Theorems}
Throughout this section we consider elliptic linear partial differential 
operators which, in local coordinates, are 
written as\footnote{In particular, it is also assumed that all operators are of second order.}
\begin{gather}
 P u = \partial_j (a^{ij} \partial_i u ) + b^i \partial_i u - c u,
\label{elliptic_operator}
\end{gather}
where the functions $a^{ij}$, $b^i$ and $c$ are assumed to be smooth, and manifolds are always
assumed to be compact and without boundary.

A differential operator is elliptic when its symbol is invertible \cite{Kaz, Au2}. For scalar equations,
this corresponds to saying that the matrix $a^{ij}$ in (\ref{elliptic_operator}) is positive
definite for all $x \in M$. In other words, there exists a constant $\la > 0$, called 
the \emph{ellipticity constant}, such that
\begin{gather}
 a^{ij}(x) \xi_i \xi_j \geq \la |\xi|^2
\nonumber
\end{gather}
for any $x \in M$ and tangent vectors $\xi\neq 0$ (where the above expression 
should be understood in local coordinates around $x$).
The standard model of an elliptic operator is the Laplacian of a metric $g$, $\Delta_g$. 

Recall that a linear system in finite dimensions, written as $A x = b$, has a solution if, and only 
if, $b$ is orthogonal to the kernel of $A^*$. Much of the power of elliptic theory stems from the fact that
elliptic operators enjoy a similar property:

\begin{prop} (Fredholm Alternative, \cite{Kaz, Evans, GT}). Let $M$ be a smooth Riemannian manifold and 
let $P: C^\infty(M) \rar C^\infty(M)$ be a linear elliptic differential operator, and denote by $P^*$ its formal
adjoint. Given $f \in C^\infty(M)$, the equation 
\begin{gather}
 P u = f
\label{equation_Fred}
\end{gather}
has a solution if, and only if, $f$ is $L^2$ orthogonal to the kernel of $P^*$. In particular, if $P^* = P$
then equation (\ref{equation_Fred}) always has a solution, which is unique and smooth, provided that $P$ has trivial kernel.
\label{Fredholm_alternative}
\end{prop}

Proposition \ref{Fredholm_alternative} can be generalized to include differential operators between spaces of sections 
of vector bundles over $M$, rough coefficients and situations of very low regularity. See the above references, 
\cite{Tr1} and references therein.

For self-adjoint operators, the Fredholm alternative reduces the problem of existence to that of uniqueness.
Uniqueness can be tackled with the maximum principle:

\begin{prop} (Maximum principle, \cite{Kaz, Evans, GT}) Consider the operator $P$ in (\ref{elliptic_operator})
on a domain $\Om \subseteq M$ (possibly $\Om = M$). \\
\noindent (i) Assume $ c = 0$ and $Pu \geq 0$ (resp. $\leq 0$) in $\Om$. If $u$ achieves its maximum (resp. minimum)
in the interior of $\Om$ then $u$ is constant.\\
\noindent (ii) Assume  $c \geq 0$ and $Pu \geq 0$ (resp. $\leq 0$) in $\Om$.
Then $u$ cannot achieve a non-negative maximum (resp. non-positive minimum) in the interior of $\Om$ unless it is constant. \\
\noindent (iii) Assume $c \geq 0$ with $c$ not identically zero. If $Pu =0$ in $M$, then $u \equiv 0$.
\end{prop}

Again, this can be generalized to include rough coefficients. Such generalizations usually involve 
replacing the previous inequalities 
by some inequality to hold in the sense of distributions. See \cite{GT, Tr2} and references therein.

\begin{prop} (Elliptic regularity). Let $u$ be a solution of $P u = f$. \\
 
\noindent (1) If $f \in C^{k,\al}(M)$, then $u \in C^{k+2,\al}(M)$. In particular $u$ is smooth 
if $f$ is smooth. \\

\noindent (2) If $f \in L^p_k(M)$, then $u \in L^p_{k+2}(M)$. 

\label{elliptic_regularity}
\end{prop}

Elliptic regularity says that solutions always ``gain two derivatives'' as compared to 
the right hand side term. This leads to the following useful bootstrap argument. Consider a non-linear
equation of the form
\begin{gather}
 P u = f(x, u),
\nonumber
\end{gather}
where $f: M \times \RR \rar \RR$ is smooth, and suppose that somehow one manages to produce a solution
$u$ which is in $C^{0,\al}$. Since the composition of a smooth function with a H\"older continuous function
is again H\"older continuous, we have $f(\cdot, u) \in C^{0,\al}(M)$. Hence elliptic regularity tells us 
that $u \in C^{2,\al}(M)$. But then $f(\cdot, u)$ is in fact in $C^{2,\al}(M)$, and so $u$ must 
be in $C^{4,\al}(M)$. Continuing this argument yields $u \in C^\infty(M)$. Of course, if $f$ is, say, 
only $C^k$, we can iterate this argument only up to order $k$.

\begin{prop} (Schauder estimates \cite{Kaz,GT,Au2})
  Let $M$ be a smooth Riemannian manifold and 
let $P: C^\infty(M) \rar C^\infty(M)$ be a linear elliptic differential operator.
There exist a constant $C$, depending only on the $\al$, the dimension of $M$,
the $C^{0,\al}$-norm of $a^{ij}$, $b^i$ and $c$, and on the ellipticity constant $\la$, such that
\begin{gather}
\parallel u \parallel_{C^{k+2,\al}(M)} \, 
\leq  C \big ( \parallel P u \parallel_{C^{k,\al}(M)} + \parallel u \parallel_{C^0(M)} \big )
\nonumber
\end{gather}
Moreover, if one restricts $u$ so that it is orthogonal (in $L^2$) to the kernel of $P$, 
then we can drop the $C^0$ term on the right hand side (after replacing $C$ with a new constant
depending only on the same quantities as before).
\label{schauder}
\end{prop}

The next proposition is a powerful tool to produce solutions to semi-linear equations, as it can be seen, for example,
from its applications in section \ref{d_not_4}. Its weakness relies on the fact that it 
does not guarantee uniqueness of solutions (and in fact it is easy to construct examples where the 
proposition applies but solutions are not unique, see \cite{Kaz}).

\begin{prop} (Sub- and super-solutions, \cite{SY1,Kaz}) Let $M$ be a smooth Riemannian 
manifold. Consider 
the semi-linear elliptic equations
\begin{gather}
 \Delta_g u + f(x,u) = 0,
\label{sub_super_sol_prop}
\end{gather}
where $f \in C^\infty(M\times\RR)$. Suppose there there exist $\phi,~\psi \in C^2(M)$ satisfying
\begin{gather}
 \Delta_g \phi + f(x,\phi) \geq 0, \nonumber \\
\Delta_g \psi + f(x,\psi) \leq 0, \nonumber
\end{gather}
(such $\phi$ and $\psi$ are called respectively a sub-solution and super-solution for (\ref{sub_super_sol_prop})), 
and $\phi \leq \psi$. Then (\ref{sub_super_sol_prop}) has a solution $u \in C^\infty(M)$ such that $\phi \leq u \leq \psi$.
 \label{sub_super_sol}
\end{prop}

Proposition \ref{sub_super_sol} admits further generalizations, including 
for manifolds with boundary; see \cite{Sat,DK}.\\

\subsection{Notation and conventions} Some notation and conventions we use: \\

\noindent Laplacian: $\Delta_g = \frac{1}{\sqrt{|g|}}\partial_i(\sqrt{|g|}g^{ij}\partial_j)$.
This differs by a sign from \cite{Do}. \\

\noindent Spacetime metric: $-++\cdots+$. \\

\noindent Dimension: spacetime before compactification: $D$, spacetime after compactification: $d$, 
compact dimensions: $n$. \\

\noindent $L^p_k(M,g)$ denotes the Sobolev space of $k^{th}$ weakly differentiable functions
which belong, along with its weak derivatives, to $L^p(M,g)$, where the measure of integration
is the volume element of the metric $g$. We sometimes write $L^p_k(M)$ for 
simplicity (in fact, when $M$ is compact, Sobolev spaces defined using different metrics are equivalent, see
for example \cite{Ro}). \\

\noindent $C^{k,\al}(M)$ denotes the H\"older spaces with $k$ derivatives and H\"older exponent $\al$.

\end{document}